\newcommand{\ket}[1]{\mbox{$ | #1 \rangle $}}

\newcommand{\bla}{\color{black}}
\newcommand{\be}{\begin{equation}}
\newcommand{\ee}{\end{equation}}
\newcommand{\ba}{\begin{eqnarray}}
\newcommand{\ea}{\end{eqnarray}}

\newenvironment{definition}[1][Definition.]{\begin{trivlist}
\item[\hskip \labelsep {\bfseries #1}]}{\end{trivlist}}

\documentclass{iopart}
\usepackage{braket}
\usepackage{cite}
\usepackage{color}
\expandafter\let\csname equation*\endcsname\relax
\expandafter\let\csname endequation*\endcsname\relax
\usepackage{amsmath}
\usepackage{amssymb}
\usepackage{amsthm} 
\usepackage{amsfonts}
\usepackage{dsfont} 
\usepackage{setspace}
\usepackage{graphicx}
\usepackage[colorlinks=true,citecolor=blue,urlcolor=blue]{hyperref} 

\newtheorem{observation}{Observation}
\newtheorem{Theorem}{Theorem}

\begin{document}

\title[Operational nonclassicality of local multipartite correlations]
{Operational nonclassicality of local multipartite correlations 
in the limited-dimensional simulation scenario}

\author{C.  Jebaratnam} \ead{jebarathinam@bose.res.in}
\address{S. N. Bose National Centre for Basic Sciences, Salt Lake,
  Kolkata      700     098,      India}     \author{Debarshi      Das}
\ead{debarshidas@jcbose.ac.in} \address{Centre for Astroparticle
  Physics and Space Science (CAPSS),  Bose Institute, Block EN, Sector
  V,  Salt Lake,  Kolkata 700  091, India}  \author{Suchetana Goswami}
\address{S.   N.  Bose  National Centre  for Basic  Sciences, Salt
  Lake,    Kolkata    700    098,   India}   
\author{R.   Srikanth}  \ead{srik@poornaprajna.org}
\address{Poornaprajna Institute of  Scientific Research Bangalore-
  560    080,   Karnataka,    India}    \author{A.    S.     Majumdar}
\ead{archan@bose.res.in} \address{S. N. Bose National Centre for
  Basic Sciences, Salt Lake, Kolkata 700 098, India} 
\vspace{10pt}
\begin{indented}
\item[]February 2014
\end{indented}

\begin{abstract}
 For a  bipartite local  quantum correlation, superlocality  refers to
 the requirement for a larger dimension  of the random variable in the
 classical simulation  protocol than that  of the quantum  states that
 generate the correlations.   In this work, we  consider the classical
 simulation of  local tripartite quantum correlations  $P$ among three
 parties  $A,  B$ and  $C$.   If  at  least  one of  the  bipartitions
 $(A|BC)$, $(B|AC)$ and $(C|AB)$ is superlocal, then $P$ is said to be
 absolutely  superlocal,   whereas  if  all  three   bipartitions  are
 superlocal, then $P$ is said  to be genuinely superlocal.  We present
 specific   examples   of   genuine   superlocality   for   tripartite
 correlations  derived from  three-qubit  states.  It  is argued  that
 genuine quantumness as  captured by the notion of  genuine discord is
 necessary  for demonstrating  genuine superlocality.     Finally,
   the notions of absolute and  genuine superlocality are also defined
   for multipartite correlations.
\end{abstract}

%
%
%
%
%

\section{Introduction}

Quantum  composite  systems  exhibit  nonclassical  features  such  as
entanglement \cite{Horodecki} and  Bell nonlocality \cite{BCP+14}.  In
the  Bell   scenario,  a   correlation  arising  from   local  quantum
measurements  on  a composite  system  is  nonlocal  if it  cannot  be
simulated by sharing classical randomness,  i.e., if it does not admit
a local  hidden variable (LHV)  model \cite{Bel64}.  It  is well-known
that  quantum  mechanics  is  not  maximally  nonlocal  as  there  are
nonsignaling  (NS)  correlations,   e.g.,  Popescu-Rohrlich  (PR)  box
\cite{PR94},  which are  more nonlocal  than that  allowed by  quantum
theory.  In  the framework of generalized  nonsignaling theory (GNST),
there have been  various attempts to find out  the minimal constraints
of  quantum theory  \cite{Pop14}. GNST  has  also been  used in  other
contexts  such as  quantum cryptography  \cite{AGM06} and  quantifying
nonlocality \cite{EPR2B}.

Recently,  it  has   been  shown  that  even   some  separable  states
demonstrate  advantage in  certain quantum  information tasks  if they
have  quantum  discord  which  is a  generalized  measure  of  quantum
correlations \cite{OZ01}.  This initiated the study of nonclassicality
going beyond nonlocality \cite{Modietal}.  In Ref. \cite{BP14}, it has
been shown  that certain separable  states which have  quantumness may
improve  certain  information  theoretic   protocols  if  there  is  a
constraint on the possible local  hidden variable models, i.e., shared
randomness between the parties is  limited to be finite. This provides
a  way  for  obtaining  an  operational meaning  of  the  measures  of
quantumness such as quantum discord .

While  entanglement is  necessary for  nonlocality, not  all entangled
states  can be  used to  demonstrate nonlocality.   In the  context of
classical  simulation   of  local   entangled  states,   Bowles  \etal
\cite{BHQ+15} have defined a measure which is the minimal dimension of
shared classical randomness that is needed to reproduce the statistics
of a  given local entangled  state.  Unlike the previous  works, which
used unbounded shared  randomness to simulate a  given local entangled
state, Bowles \etal have shown that  all local entangled states can be
simulated by using only finite shared randomness.

In Ref. \cite{DW15}, Donohue and Wolfe (DW) have provided upper bounds
on the minimal  dimension of the shared  classical randomness required
to simulate any local correlation in  a given Bell scenario.  
DW  have demonstrated  an interesting  feature of  certain local
boxes  called superlocality:  there  exist local  boxes  which can  be
simulated by  certain quantum  systems of  local dimension  lower than
that of the shared classical randomness needed to simulate them.  That
is, superlocality  refers to  the dimensional advantage  in simulating
certain  local boxes  by  using quantum  systems.   In particular,  DW
\cite{DW15}  and   also  Goh   \etal  \cite{GBS16}  have   shown  that
entanglement enables superlocality, however, superlocality occurs even
for separable states.

In  Ref.   \cite{JAS16}, Jebaratnam \etal have pointed  out  that superlocality  
cannot   occur  for  arbitrary  separable   states.  In particular, 
Jebaratnam \etal have argued that  separable states which  are a
classical-quantum state \cite{PHH08} or its permutation can never lead
to superlocality.  Note that bipartite  quantum states which are not a
classical-quantum state must have quantumness as quantified by quantum
discord  \cite{OZ01}.  The  observation that  the limited  dimensional
quantum simulation of certain  local correlations requires quantumness
in the states motivated the study of nonclassicality going beyond Bell
nonlocality \cite{Jeba}.  In Ref.  \cite{JAS16}, Jebaratnam \etal
have considered a
measure of nonclassicality in the context of GNST called Bell strength
for  a  family  of  local correlations  and  identified  nonzero  Bell
strength as a quantification of superlocality  as well. 

The extension of  the Bell-type scenario to more than  two parties was
first elaborated  by Greenberger, Horne, and  Zeilinger \cite{GHZ}. In
Refs.  \cite{SI,mermin,PM12,DPM13},  certain interesting  features  of
nonlocality of  tripartite systems  over and above  what is  found for
bipartite ones  have been  established. Recently,  genuine multipartite
quantum discord  has been defined  to quantify the  quantumness shared
among  all subsystems of  the multipartite quantum  state
\cite{GTC,MZ12,GCTS,GTQD}.    In   Ref.    \cite{Jeb17},   Jebaratnam   has
demonstrated  that  the  limited  dimensional  quantum  simulation  of
certain local tripartite correlations  must require genuine tripartite
quantum  discord states.   To study  genuine nonclassicality  of these
correlations,  Jebaratnam   has  introduced  two   quantities  called,
Svetlichny strength and Mermin strength,  in the context of tripartite
nonsignaling boxes (i.e., correlations in the tripartite GNST).

For the case  of multipartite systems, early  studies have established
that certain nonlocal measures may indeed be amplified by the addition
of   system  dimensions   \cite{Mer90,RS91,HM95,Cab02}.   Multipartite
quantum entanglement,  however, displays complicated  structures which
can be broadly classified according  to whether entanglement is shared
among all  the subsystems of a  given multipartite system or  not.  In
Ref.   \cite{SI},   Svetlichny  introduced   the  notion   of  genuine
multipartite nonlocality and derived  Bell-type inequalities to detect
it.  Quite  generally, \textit{genuine} multipartite nonlocalty  for a
multipartite  system  occurs  precisely  when  every  bi-partition  is
nonlocal.    By  contrast,  \textit{absolute} nonlocality  occurs
  when at least one  bi-partition is nonlocal.   Absolute multipartite
nonlocality  is indicated  by  the violation  of  a Mermin  inequality
\cite{mermin}.

In the  present work,  we are interested  in providing  an operational
characterization   of  genuine   and   absolute  nonclassicality   via
superlocality in  the scenario of  local tripartite systems  and
  its generalization  to higher-particle systems.  It  is argued that
genuine  quantum  discord  is   necessary  for  demonstrating  genuine
superlocality.  This  implies that  genuine superlocality  provides an
operational characterization of genuine  quantum discord.  The concept
of superlocality  is generalized  for multipartite  correlations also.
As specific  examples, we  consider tripartite correlations  that have
non-vanishing Svetlichny strength or Mermin strength \cite{Jeb17}.

The organization of  the paper is as follows.  In  Sec.  \ref{prl}, we
present the mathematical tool that we  use for the purpose of studying
superlocality  of  local  tripartite   boxes,  namely  a  polytope  of
tripartite             nonsignaling             boxes             with
two-binary-inputs-two-binary-outputs.    In  Sec.    \ref{gslsec},  we
define absolute and genuine superlocality  for tripartite boxes and we
present examples of two families  of tripartite local correlations that are
absolutely   superlocal   or   genuinely  superlocal,   having   their
nonclassicality  quantified  in  terms of  nonzero  \textit{Svetlichny
  strength}  and nonzero  \textit{Mermin strength},  respectively.  In
Sec.   \ref{gslgqd}, we  demonstrate  the  connection between  genuine
superlocality and  genuine quantum  discord.  In Sec.   \ref{gmnl}, we
generalize the  notion of  superlocality to  $n$-partite correlations.
Sec.  \ref{conc} is reserved for certain concluding remarks.

\section{Preliminaries \label{prl}}

We are interested  in quantum correlations arising  from the following
tripartite  Bell scenario.   Three  spatially  separated parties,  say
Alice, Bob,  and Charlie, perform  three dichotomic measurements  on a
shared              tripartite              quantum              state
$\rho\in\mathcal{H}_A \otimes \mathcal{H}_B\otimes\mathcal{H}_C$,  where
$\mathcal{H}_K$  denotes  Hilbert  space  of  $k$th  party.   In  this
scenario, a correlation  between the outcomes is described  by the set
of conditional probability distributions $P(abc|xyz)$, where $x$, $y$,
and $z$ denote  the inputs (measurement choices) and $a$,  $b$ and $c$
denote the  outputs (measurement outcomes)  of Alice, Bob  and Charlie
respectively (with $x,y,z,a,b,c\in  \{0,1\}$).  Suppose $M^{a}_{A_x}$,
$M^{b}_{B_y}$ and  $M^{c}_{C_z}$ denote  the measurement  operators of
Alice, Bob,  and Charlie, respectively.   Then the correlation  can be
expressed    through    the    Born's    rule    as    follows:    \be
P(abc|xyz)=\mathrm{Tr}          \left(\rho          M^{a}_{A_x}\otimes
M^{b}_{B_y}\otimes M^{c}_{C_z}\right).  \ee In particular, we focus on
the scenario where  the parties share a three-qubit  state and perform
spin     projective     measurements     $A_x=\hat{a}_x.\vec{\sigma}$,
$B_y=\hat{b}_y.\vec{\sigma}$,  and $C_z=\hat{c}_z.\vec{\sigma}$.  Here
$\hat{a}_x$,  $\hat{b}_y$,  and  $\hat{c}_z$ are  unit  Bloch  vectors
denoting         the         measurement        directions         and
$\vec{\sigma}=\{\sigma_1,\sigma_2,\sigma_3\}$,                    with
$\{\sigma_i\}_{i=1,2,3}$ being the Pauli matrices.

The set of nonsignaling (NS) boxes with two binary inputs and two binary outputs
forms a convex polytope $\mathcal{N}$ in a $26$ dimensional space \cite{BLM+05} 
and includes the set of quantum correlations $Q$ as a proper subset. 
Any box belonging to this polytope can be fully specified by $6$ single-party, 
$12$ two-party and $8$ three-party expectations,
\begin{align}
&P(abc|xyz)=\frac{1}{8}[1+(-1)^a\braket{A_x}+(-1)^b\braket{B_y}+(-1)^c\braket{C_z}\nonumber \\
&+(-1)^{a\oplus b}\braket{A_xB_y}+(-1)^{a\oplus c}\braket{A_xC_z}+(-1)^{b\oplus c}\braket{B_yC_z}\nonumber \\
&+(-1)^{a\oplus b\oplus c}\braket{A_xB_yC_z}],
\end{align}
where $\braket{A_x}=\sum_a (-1)^a P(a|x)$, $\braket{A_xB_y}=\sum_{a,b}(-1)^{a\oplus b}P(ab|xy)$ 
and $\braket{A_xB_yC_z}=\sum_{a,b,c}(-1)^{a\oplus b \oplus c}P(abc|xyz)$, $\oplus$ denotes modulo sum $2$.
The set of boxes that can be simulated by a fully LHV model are of the form,
\begin{align}
P(abc|xyz)=\sum^{d_\lambda-1}_{\lambda=0} p_\lambda P_\lambda(a|x)P_\lambda(b|y)P_\lambda(c|z), \label{LHV}
\end{align}
with $\sum^{d_\lambda-1}_{\lambda=0} p_\lambda=1$ and it forms the fully local (or, $3$-local) 
polytope \cite{Fin82,WernerWolfmulti} denoted by $\mathcal{L}$. 
Here $\lambda$ denotes shared classical randomness which occurs with probability $p_\lambda$. 
For a given fully local box, the form (\ref{LHV}) determines a 
classical simulation protocol with dimension $d_\lambda$ \cite{DW15}.
The extremal boxes of $\mathcal{L}$ are $64$ local vertices which are fully deterministic boxes,
\be
P^{\alpha\beta\gamma\epsilon\zeta\eta}_D(abc|xyz)=\left\{
\begin{array}{lr}
1, & a=\alpha x\oplus \beta\\
   & b=\gamma y\oplus \epsilon \\
   & c=\zeta  z\oplus \eta\\
0 , & \text{otherwise}.\\
\end{array}
\right.  \label{DB} \ee
Here, $\alpha,\beta,\gamma, \epsilon, \zeta, \eta \in\{0,1\}$. The above boxes can be written as the 
product of deterministic distributions corresponding to Alice and Bob-Charlie,  i.e.,
$P^{\alpha\beta\gamma\epsilon\zeta\eta}_D(abc|xyz)=P^{\alpha\beta}_D(a|x)P^{\gamma\epsilon\zeta\eta}_D(bc|yz)$,
here 

\be
\label{alicedet}
P^{\alpha\beta}_D(a|x)=\left\{
\begin{array}{lr}
1, & a=\alpha x\oplus \beta\\
0 , & \text{otherwise}\\
\end{array}
\right.   \ee
and
\be
P^{\gamma\epsilon\zeta\eta}_D(bc|yz)=\left\{
\begin{array}{lr}
1, & b=\gamma y\oplus \epsilon \\
   & c=\zeta  z\oplus \eta\\
0 , & \text{otherwise},\\
\end{array}
\right.   \ee
which can also be written as the product deterministic distributions corresponding to Bob and Charlie, 
i.e., $P^{\gamma\epsilon\zeta\eta}_D(bc|yz)=P^{\gamma\epsilon}_D(b|y)P^{\zeta\eta}_D(c|z)$,
where
\be
P^{\gamma\epsilon}_D(b|y)=\left\{
\begin{array}{lr}
1, & b=\gamma y\oplus \epsilon\\
0 , & \text{otherwise}\\
\end{array}
\right.   \ee
and 
\be
P^{\zeta\eta}_D(c|z)=\left\{
\begin{array}{lr}
1, & c=\zeta  z\oplus \eta\\
0 , & \text{otherwise}.\\
\end{array}
\right.  \ee  Note that the set  of $3$-local boxes and  quantum boxes
satisfy  $\mathcal{L} \subset  Q  \subset  \mathcal{N}$.  Boxes  lying
outside  $\mathcal{L}$ are  called \textit{absolutely}  nonlocal boxes
and  they  cannot  be  written  as  a  convex  mixture  of  the  local
deterministic boxes alone.

A tripartite  correlation is  said to  be ``2-local  across the
  bipartite cut  $(AB|C)$'' if it has the following form:
  \begin{align}
P(abc|xyz)&=\sum_\lambda p_\lambda P_\lambda(ab|xy)\,P_\lambda(c|z),
\label{HLNLbi}
\end{align}
where $P_\lambda(ab|xy)$ can have  arbitrary nonlocality  consistent with  the NS
principle.  ``2-locality across other bipartite cuts" for tripartite correlations 
can be  defined similarly. Note that any $3$-local box can always be written in 
$2$-local form across any possible bipartition. The general  \textit{2-local} form
\cite{Banceletal} is, therefore:
\begin{align}
P(abc|xyz)&=s_1\sum_\lambda p_\lambda P_\lambda^{AB|C}+s_2\sum_\lambda q_\lambda P_\lambda^{AC|B}\nonumber \\
&+s_3\sum_\lambda r_\lambda P_\lambda^{A|BC}, \label{HLNL}
\end{align}
where $P_\lambda^{AB|C}=P_\lambda(ab|xy)\,P_\lambda(c|z)$,  and, where
$P_\lambda^{AC|B}$ and $P_\lambda^{A|BC}$  are similarly defined. Here
$s_1  + s_2  +  s_3 =1$;  $\sum_\lambda  p_\lambda =1$;  $\sum_\lambda
q_\lambda  =1$;  $\sum_\lambda  r_\lambda  =1$.   
 Each bipartite  distribution  in the  decomposition
(\ref{HLNL})  can have  arbitrary nonlocality  consistent with  the NS
principle.
A  tripartite
  nonlocal  box is  \textit{genuinely  tripartite nonlocal}  if it 
  cannot  be written  in  the $2$-local  form given by Eq.
  (\ref{HLNL}). Hence, a genuinely tripartite nonlocal box is nonlocal
  with respect to every bipartition $(A|BC)$, $(B|CA)$, $(C|AB)$.

The set  of boxes that admit  a decomposition as in  Eq.  (\ref{HLNL})
again  forms  a  convex   polytope  which  is  called  \textit{2-local
  polytope}, denoted  by $\mathcal{L}_2$.  The extremal  boxes of this
polytope  are the  $64$ local  vertices and  $48$ $2$-local  vertices.
There are  $16$ $2$-local  vertices in which  a PR-box  \cite{PR94} is
shared between $A$ and $B$,
\begin{align}
&P^{\alpha\beta\gamma\epsilon}_{12}(abc|xyz)\nonumber \\
&=\left\{
\begin{array}{lr}
\frac{1}{2}, & a\oplus b=x\cdot y \oplus \alpha x\oplus \beta y \oplus \gamma \quad \& \quad c=\gamma z \oplus \epsilon\\
0 , & \text{otherwise},\\
\end{array}
\right. \label{PR}
\end{align}
and      the     other      $32$      two-way     local      vertices,
$P^{\alpha\beta\gamma\epsilon}_{13}$                               and
$P^{\alpha\beta\gamma\epsilon}_{23}$, in  which a PR-box is  shared by
$AC$  and $BC$,  are similarly  defined.   The extremal  boxes in  Eq.
(\ref{PR})    can    be    written    in    the    factorized    form,
$P^{\alpha\beta\gamma\epsilon}_{12}             (abc|xyz)            =
P^{\alpha\beta\gamma}_{PR} (ab|xy) P^{\gamma\epsilon}_D(c|z)$, where
\begin{align}
&P_{PR}^{\alpha \beta \gamma} (ab|x y)\nonumber \\
&=\left\{
\begin{array}{lr}
\frac{1}{2}, & a \oplus b = x.y \oplus \alpha x \oplus \beta y \oplus \gamma\\
0 , & \text{otherwise},\\
\end{array}
\right. \label{PR2}
\end{align}
and  $P^{\gamma\epsilon}_D(c|z)$ is  the aforementioned  deterministic
box.   The  set  of  $2$-local  boxes  satisfy,  $\mathcal{L}  \subset
\mathcal{L}_2  \subset  \mathcal{N}$.   Absolute  tripartite  nonlocal
boxes  can be  either  2-local or  genuinely  tripartite nonlocal.   A
genuinely  tripartite  nonlocal box  cannot  be  written as  a  convex
mixture  of the  extremal boxes  of $\mathcal{L}_{2}$  and violates  a
facet inequality of $\mathcal{L}_2$ given in Ref. \cite{Banceletal}.

The  Svetlichny   inequalities  \cite{SI}  which  are   given  by  \be
\mathcal{S}_{\alpha\beta\gamma\epsilon}    =\sum_{xyz}(-1)^{x\cdot   y
  \oplus x\cdot z \oplus y\cdot z \oplus \alpha x\oplus \beta y \oplus
  \gamma z \oplus \epsilon}\braket{A_xB_yC_z}\le4, \label{SI} \ee 
are one of the classes of  facet inequalities of the 2-local polytope.
The violation of  a Svetlichny inequality implies one of  the forms of
genuine tripartite nonlocality \cite{Banceletal}. 
The following extremal tripartite
nonlocal boxes:
\begin{align}
&P^{\alpha\beta\gamma\epsilon}_{\rm Sv}(abc|xyz) \nonumber \\
&=\left\{
\begin{array}{lr}
\frac{1}{4}, & \!a\!\oplus \!b\!\oplus \!c\!
=\!x\cdot y \!\oplus \!x\cdot z\! \oplus \!y\cdot z \!\oplus \!\alpha x\!\oplus\! \beta y\! \oplus\! \gamma z \!\oplus\! \epsilon\\
0 , & \text{otherwise},\\
\end{array}
\right. \label{NLV}
\end{align}
which violate a Svetlichny inequality to its algebraic maximum are called Svetlichny boxes.

Mermin  inequalities \cite{mermin}  are one  of the  classes of  facet
inequalities of  the fully  local (or, $3$-local) polytope \cite{PS01,Sli03}.   One of
the      Mermin      inequalities      is      given      by,      \be
\braket{A_0B_0C_0}-\braket{A_0B_1C_1}-\braket{A_1B_0C_1}-\braket{A_1B_1C_0}\le2, \label{MI}
\ee and  the other $15$ Mermin  inequalities can be obtained  from the
above inequality by local reversible operations which are analogous to
local  unitary   operations  in  quantum  theory   and  include  local
relabeling of the inputs  and outputs. All the $16$ Mermin inequalities
are given by 
\be \label{LROMI}
\mathcal{M}_{\alpha\beta\gamma\epsilon}=
(\alpha\oplus\beta\oplus\gamma\oplus1)\mathcal{M}^+_{\alpha\beta\gamma\epsilon}+(\alpha\oplus\beta\oplus\gamma)\mathcal{M}^-_{\alpha\beta\gamma\epsilon}\le2, 
\ee
where $\mathcal{M}^+_{\alpha\beta\gamma\epsilon}:=(-1)^{\gamma\oplus\epsilon}\braket{A_0B_0C_1}
+(-1)^{\beta\oplus\epsilon}\braket{A_0B_1C_0}+(-1)^{\alpha\oplus\epsilon}\braket{A_1B_0C_0}
+(-1)^{\alpha\oplus\beta\oplus\gamma\oplus\epsilon\oplus1}\braket{A_1B_1C_1}$ and\\ $\mathcal{M}^-_{\alpha\beta\gamma\epsilon}:=(-1)^{\alpha\oplus\beta\oplus\epsilon\oplus 1}
\braket{A_1B_1C_0}+(-1)^{\alpha\oplus\gamma\oplus\epsilon\oplus 1}\braket{A_1B_0C_1}
+(-1)^{\beta\oplus\gamma\oplus\epsilon\oplus 1}\braket{A_0B_1C_1}+(-1)^{\epsilon}\braket{A_0B_0C_0}$. 
Mermin inequalities detect
absolute  nonlocality, i.e.,  it  guarantees the  box  to lie  outside
$\mathcal{L}$\bla.   Quantum   correlations  that  violate   a  Mermin
inequality      to     its      algebraic     maximum      demonstrate
Greenberger--Horne--Zeilinger (GHZ) paradox  \cite{GHZ} and are called
Mermin boxes.

In  the  tripartite case,  while  the  dimension  of the  NS  polytope
$\mathcal{N}$  is  26,   the  number  of  extreme   boxes  are  53,856
\cite{PBS11}.  Thus, a given $P$, expressed as a convex combination of
the extreme  boxes, does not  have a unique decomposition.   The
  correlations  $P$ of  interest to  us,  belong to  a subpolytope  of
  $\mathcal{N}$, called Svetlichny-box polytope in Ref.  \cite{Jeb17},
  denoted  $\mathcal{R}$,  having $128$  extreme  boxes  (64 3-local
  boxes,  48  2-local  boxes  that are not  3-local,  and  finally,  16
  Svetlichny boxes).  The Svetlichny-box polytope  can be seen as
  a generalization of the bipartite  PR-box polytope. 

  Even in  the case of the polytope $\mathcal{R}$,  since it has a
  more extreme  boxes than its  dimension of  26, therefore it  has no
  unique  decomposition.   However,    a  unique,  \textit{canonical},
decomposition for any $P\in\mathcal{R}$ can be given, as shown in Ref.
\cite{Jeb17},  by making  use of  the non-trivial  symmetry properties
among the extreme  boxes of $\mathcal{R}$.  Thereby,  any nonlocal box
$P\in\mathcal{R}$ violating a Svetlichny  inequality can be brought to
the form \cite{Jeb17}:
\begin{equation}
P=p_{Sv}P^{\alpha\beta\gamma\epsilon}_{Sv}+(1-p_{Sv})P_{SvL},  
\label{eq:canonical}
\end{equation} 
where  a single  Svetlichny-box, $P^{\alpha\beta\gamma\epsilon}_{Sv}$,
is  dominant  and $p_{Sv}$  has  been  maximized  and $P_{SvL}$  is  a
Svetlichny-local box.  

Eq. (\ref{eq:canonical}) is called the canonical decomposition for any
box  $P$   belonging  to   the  Svetlichny-box   polytope.   Following
Ref. \cite{Jeb17}, we refer to $p_{Sv}$ in the canonical decomposition
as  Svetlichny   strength.      In  a   given  situation,  suppose
  $\textbf{P}^\uparrow_{\rm Sv}$ denotes  the dominant Svetlichny-box.
  The canonical form (\ref{eq:canonical}) becomes
\begin{equation}    P=\mu
  \textbf{P}^\uparrow_{\rm
    Sv}+(1-\mu)P^{\mathcal{G}=0}_{SvL}.   \label{eq:psvl}  
\end{equation}   
The   fact   that   $\mu$   represents   the   maximal   fraction   of
$\textbf{P}^\uparrow_{\rm  Sv}$   over  all  decompositions   has  the
following   consequence.    The    quantity   $\mathcal{G}$   in   Eq.
(\ref{eq:psvl}) is given by
\begin{equation}
\mathcal{G}:=\min\{\mathcal{G}_1,...,\mathcal{G}_9\},  \label{GBD} 
\end{equation}
where
$\mathcal{G}_1:=\Big|\Big||\mathcal{S}_{000}-\mathcal{S}_{001}|-|\mathcal{S}_{010}-\mathcal{S}_{011}|\Big|
-\Big||\mathcal{S}_{100}-\mathcal{S}_{101}|-|\mathcal{S}_{110}-\mathcal{S}_{111}|\Big|\Big|$
and the other  eight $\mathcal{G}_i$ are similarly defined  and can be
obtained   by  interchanging   $\mathcal{S}_{\alpha\beta\gamma}$'s  in
$\mathcal{G}_1$.     The  motivation  for this  quantity  is  the
  ``monoandrous'' nature of Svetlichny boxes, whereby they are maximal
  for precisely one of the Svetlichny inequalities. On the other hand,
  the deterministic boxes have the same absolute values.  As a result,
  $\mathcal{G}$  can  be  used  to witness  the  canonicality  of  the
  decomposition, by  checking that the local  part has $\mathcal{G}=0$
  \cite{Jeb17}.

The  Svetlichny strength $\mu$  of the box  given by
Eq. (\ref{eq:psvl})  satisfies the  relationship $\mathcal{G}(P)=8\mu$
as shown  in Ref.  \cite{Jeb17}. The  box $P^{\mathcal{G}=0}_{SvL}$ in
Eq.(\ref{eq:psvl}) is a Svetlichny-local  box having $\mathcal{G}(P) =
0$.    Thus,  the   canonical  decomposition  (\ref{eq:psvl})  is
  irreducible in the sense that  the full weight of genuine tripartite
  nonlocality  has been  transferred to  the $\textbf{P}^\uparrow_{\rm
    Sv}$ part and  no further reduction of this weight  is possible in
  the local part.

Analogously, it can  be shown (see Section  \ref{sec:mermin}) that the
canonical decomposition  of $P_{SvL}$ in Eq.   (\ref{eq:psvl}) wherein
$p_{Sv}=0$ is one  which is a convex combination of  a dominant Mermin
box (an equal mixture of two  Svetlichny boxes) and a Mermin-local box
(one that doesn't violate the Mermin inequality), such that the weight
of   the  dominant   Mermin   box  has   been  maximized.    Following
Ref. \cite{Jeb17}, we define this  maximized weight as Mermin strength
of the correlation.  

 Accordingly,  the Svetlichny-local part in  Eq. (\ref{eq:psvl}),
  which is multiply decomposable, can itself be canonically decomposed
  such   that   the    weight   of   the   Mermin    part   has   been
  maximized.  Therefore, the  boxes  that we  consider  in this  work
belong to the family of boxes which have the following canonical form:
\begin{equation}
P=\mu \textbf{P}^\uparrow_{\rm Sv}+\nu \textbf{P}^\uparrow_{\rm M}+(1-\mu-\nu)P^{\mathcal{G}=\mathcal{Q}=0}_{SvL}, \label{3dfact1}
\end{equation}
where $\textbf{P}^\uparrow_{\rm  M}$ is the dominant  Mermin-box.  The
quantity $\mathcal{Q}$  in Eq. (\ref{3dfact1}),   which witnesses
  the  canonicality  of  the  Svetlichny-local  decomposition  into  a
  Mermin-nonlocal  and Mermin-local  part,  is given  by \be  Q:=\min
\{Q_1,...,Q_9\},
\end{equation}
where
$Q_1=|||\mathcal{M}_{000}-\mathcal{M}_{001}|-|\mathcal{M}_{010}-\mathcal{M}_{011}||-||\mathcal{M}_{100}-\mathcal{M}_{101}|
-|\mathcal{M}_{110}-\mathcal{M}_{111}|||$,   and  other   $Q_i$'s  are
obtained by permutations  (Here, $\mathcal{M}_{\alpha\beta\gamma}$ are
given by  Eq. (\ref{LROMI})).   The Mermin strength  $\nu$ of  the box
given    by   Eq.    (\ref{3dfact1})   satisfies    the   relationship
$\mathcal{Q}(P)=4\nu$  as   shown  in   Ref.  \cite{Jeb17}.   The  box
$P^{\mathcal{G}=\mathcal{Q}=0}_{SvL}$ in  Eq.(\ref{3dfact1}) is  a box
having $\mathcal{G}(P) = 0$ and $Q(P) = 0$.
Unlike  the nonlocal cost \cite{EPR2B}, Svetlichny
strength and/or Mermin strength can  also be nonzero for certain fully
local as well as 2-local correlations.

In the  following, we  study the non-classicality  as captured  by the
notion of  superlocality of  two families of  fully local  and $2$-local
tripartite  correlations which  have  Svetlichny  strength and  Mermin
strength, respectively.

\section{Superlocality of tripartite fully local and $2$-local boxes}\label{gslsec}

For a given local bipartite or $n$-partite box,
let $d_\lambda$ denote the minimal dimension of the shared classical randomness. 
In Ref. \cite{DW15}, DW have derived the upper on $d_\lambda$
with the assumption that $n-1$ of the parties' measurements depend
deterministically on $\lambda$, whereas the other party can use nondeterministic
strategy on each $\lambda$ in the classical simulation model.  
Before we define superlocality for fully local tripartite boxes,
let us define superlocality for local bipartite boxes.
\begin{definition}
Suppose       we       have       a       quantum       state       in
$\mathbb{C}^{d^A}\otimes\mathbb{C}^{d^B}$   and   measurements   which
produce  a local  bipartite box  $P(ab|xy)$.  The  correlation $P$  is
\textit{sublocal} precisely if
\begin{equation}  P(ab|xy)=\sum^{d_\lambda-1}_{\lambda=0}
p_\lambda  P_\lambda(a|x)P_\lambda(b|y),  
\end{equation} 
with $d_\lambda\le d$,  with $d=\min\{d_A,d_B\}$.  Then, superlocality
holds iff there is no sublocal  decomposition of the given box.  Here,  
$\sum^{d_\lambda-1}_{\lambda=0}   p_\lambda =1$.\hfill $\blacksquare$
\end{definition}
Note that, suppose  the given local box  is produced by a  quantum state in
$\mathbb{C}^{d_A}\otimes\mathbb{C}^{d_B}$.   Then   it  can   also  be
produced by a quantum  state in $\mathbb{C}^d\otimes\mathbb{C}^d$ with
$d=\min\{d_A,d_B\}$ \cite{DW15}.

  Note that a bipartite correlation that is nonlocal is obviously
  non-sublocal  because  there  is no  dimensionally  bounded  preshared
  randomness  that   can  reproduce  it  in   a  classical  simulation
  protocol.  The interesting  cases  where  superlocality can  witness
  quantumness will then  pertain only to local  correlations. Hence in
  this  work,  we shall  be  mainly  concerned with  superlocality  as
  applied to local correlations.

 We can now extend superlocality to a tripartite system.
\begin{definition}
Suppose  we   have  a  quantum  state   in  $\mathbb{C}^{d^A}  \otimes
\mathbb{C}^{d^B}  \otimes  \mathbb{C}^{d^C}$  and  measurements  which
produce a fully local tripartite box $P(abc|xyz)$.  The  correlation $P$ is
\textit{3-sublocal} (or,  \textit{fully sublocal}) precisely if
\begin{equation}  
P(abc|xyz)= \sum^{d_\lambda-1}_{\lambda=0}
p_\lambda  P_\lambda(a|x)P_\lambda(b|y)P_\lambda(c|z),  
\label{eq:3floc}
\end{equation} 
with $d_\lambda\le d$, with  $d=\min\{d_A,d_B, d_C\}$.  Then, absolute
tripartite   superlocality   holds   iff  there   is   no   3-sublocal
decomposition  of the  given fully local box.   Here,
$\sum^{d_\lambda-1}_{\lambda=0} p_\lambda = 1$.\hfill $\blacksquare$
\end{definition}

In other words, if a fully local correlation $P(abc|xyz)$ is such that
at least  one of the  bipartitions $(A|BC)$, $(B|AC)$ and  $(C|AB)$ is
superlocal,  then $P(abc|xyz)$  is  said to  have absolute  tripartite
superlocality.  The set of all fully sublocal tripartite boxes forms a
nonconvex  subset of  the fully  local polytope  $\mathcal{L}$ and  is
denoted by $ \mathcal{L}^\ast_3 \equiv \mathcal{L}^\ast$.

\begin{figure}
\centering
\includegraphics[width=0.45\textwidth]{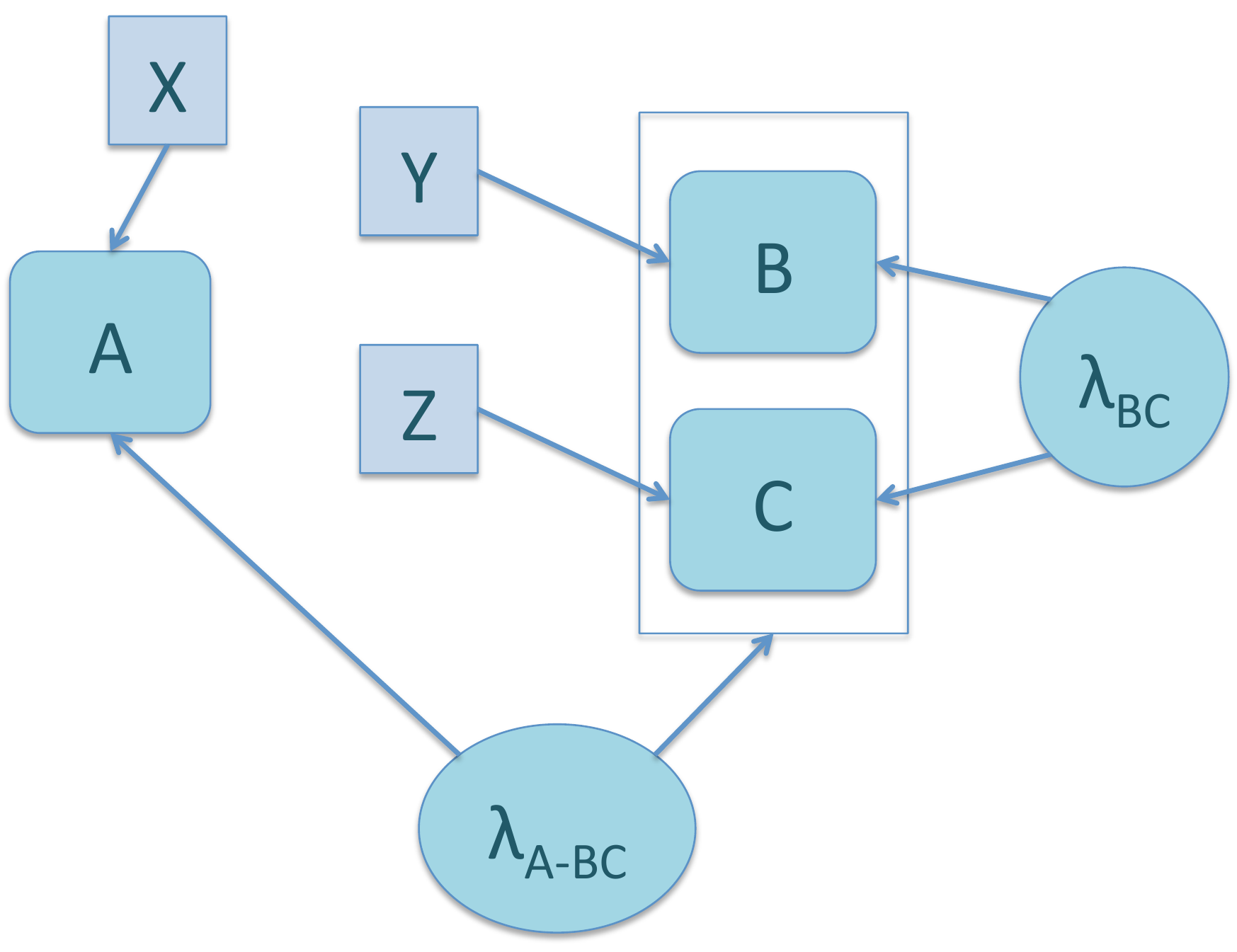}
\caption{\textbf{A   2-sublocal   box:}  Directional   acyclic   graph
  illustrating simulation of  a local tripartite box  in the following
  classical  protocol: Alice  shares hidden  variable $\lambda_{A-BC}$
  with Bob-Charlie.   For  each $\lambda_{A-BC}$, Bob  and Charlie
  can    use   shared    randomness   $\lambda_{BC}$    of   arbitrary
  dimension.}\label{gsl}
\end{figure}

We  now  define the concept of \textit{2-sublocality} across a given bipartite cut.
\begin{definition}
Suppose  we   have  a  quantum  state   in  $\mathbb{C}^{d^A}  \otimes
\mathbb{C}^{d^B} \otimes \mathbb{C}^{d^C}$  and measurements producing
a tripartite box  $P(abc|xyz)$, which is 2-local  across the bipartite
cut $(A|BC)$.  Then, the  correlation $P$ is \textit{2-sublocal across
  the bipartite cut $(A|BC)$} precisely if
\begin{align}
P(abc|xyz)&=\sum_{\lambda=0}^{d_\lambda-1} p_\lambda P_\lambda(a|x)\,P_\lambda(bc|yz), 
\label{eq:2sublocal}
\end{align}
 where  $ \sum_{\lambda=0}^{d_\lambda-1}  p_\lambda  = 1$;  $d_\lambda
 \leq \min\{d_A, d_Bd_C\}$.  Here,  $P_\lambda(bc|yz)$ is an arbitrary
 NS  box (If  $P_\lambda(bc|yz)$ is  local box,  then $P(abc|xyz)$  is
 fully local).    Tripartite  superlocality across  the bipartite
   cut $(A|BC)$ holds iff there  is no 2-sublocal decomposition across
   the bipartite cut $(A|BC)$. \hfill $\blacksquare$
\end{definition}
Tripartite superlocality  across other possible bipartite  cuts can be
defined in a  similar way.  The nonconvex set of  all 2-sublocal boxes
across different possible  bipartite cuts denoted $\mathcal{L}_2^\ast$
satisfies $\mathcal{L}^\ast \subset \mathcal{L}^\ast_2$.

Now we define genuine tripartite superlocality.
\begin{definition}
A  tripartite local  (fully  local or  $2$-local)  correlation box  is
called \textit{genuinely}  superlocal iff it is  superlocal across all
possible bipartitions.  \hfill $\blacksquare$
\end{definition}

A  fully local  correlation $P(abc|xyz)$  that isn't  $3$-sublocal has
\textit{absolute} superlocality.    In such a  correlation, there
  is  at least  one partition  for which  2-sublocality doesn't  hold.
Obviously, a fully local correlation that is genuinely superlocal must
be  absolutely  superlocal also.   Here  we  may remark  that  genuine
superlocality can occur  for absolutely nonlocal (i.e.,  which are not
fully local), but $2$-local correlations.

Two different  classical simulation protocols of  tripartite $2$-local
boxes (they may be fully local as well) are depicted in Fig. \ref{gsl}
and Fig.  \ref{sl}. Fig.  \ref{gsl} is  applicable when  the $2$-local
tripartite  box   is  fully  local   as  well;  on  the   other  hand,
Fig. \ref{sl} is  applicable when the $2$-local tripartite  box may or
may not be  fully local. We also use the  notation $\lambda$ to denote
$\lambda_{A-BC}$ for our convenience in calculations.
\begin{figure}
\centering
\includegraphics[width=0.50\textwidth]{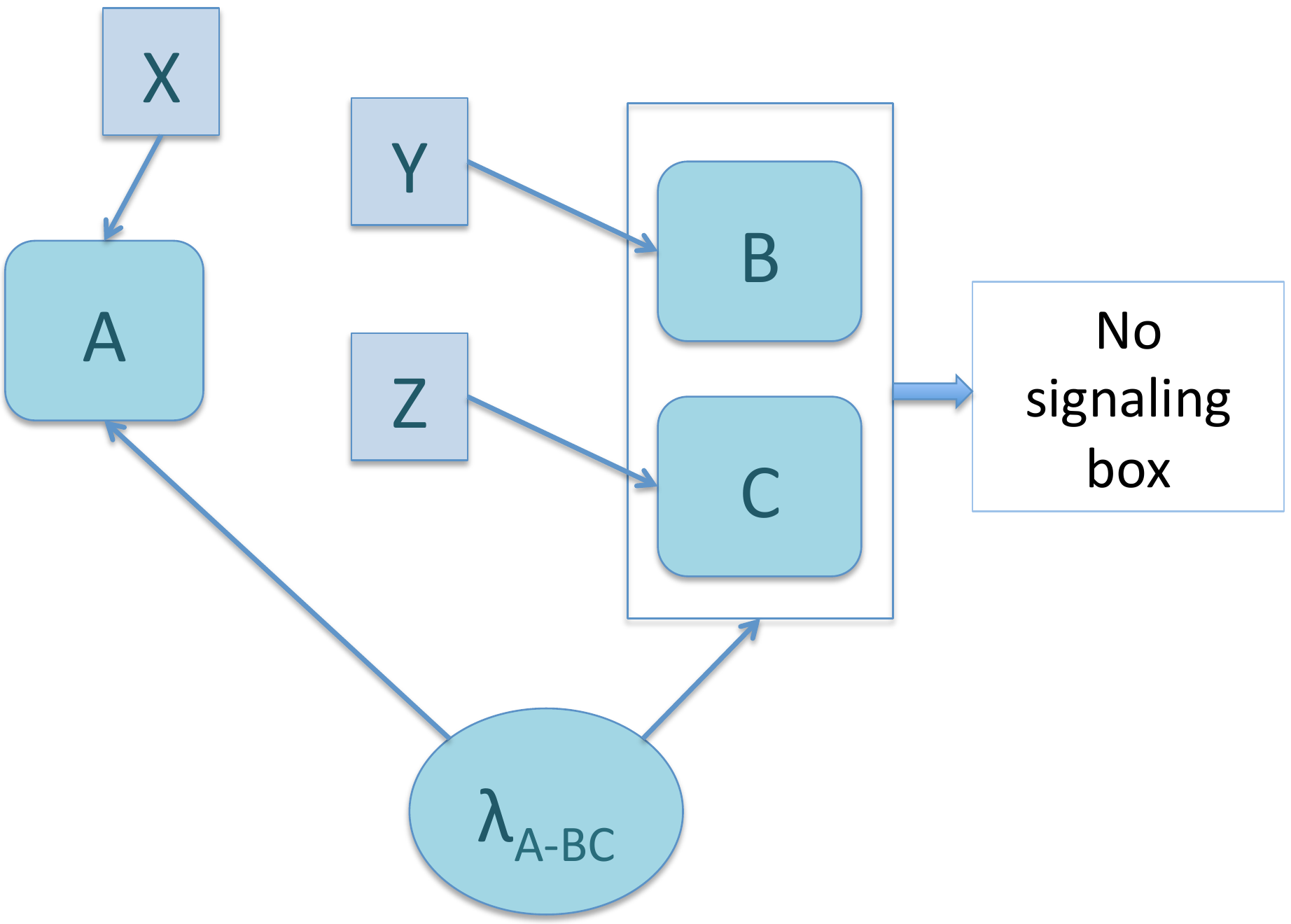}
\caption{Directional  acyclic  graph   illustrating  simulation  of  a
  2-local tripartite  box in  the following classical  protocol: Alice
  shares hidden variable $\lambda_{A-BC}$  with Bob-Charlie.   For
  each $\lambda_{A-BC}$,  the correlations of  Bob and Charlie  are NS
  boxes.}\label{sl}
\end{figure}

\subsection{Noisy Svetlichny-box}

The   generalized  GHZ   (GGHZ)  state   in  $\mathbb{C}^{2}   \otimes
\mathbb{C}^{2} \otimes \mathbb{C}^{2}$, \be
\label{gghz}
\ket{\psi_{GGHZ}}=\cos\theta\ket{000}+\sin\theta\ket{111},
\end{equation}
gives rise  to the   Svetlichny  family of  quantum correlations,
defined by:
\begin{equation}
P^{\mu}_{SvF} =  \frac{2+(-1)^{a\oplus b  \oplus c \oplus  xy\oplus xz
    \oplus yz }\sqrt{2}\mu}{16}; \quad 0<\mu\le 1 \label{SvF}
\end{equation}
for suitable choice of projective measurements.
The  above  box  violates  the Svetlichny  inequality  (given  in  Eq.
(\ref{SI}))      for       $\mu>\frac{1}{\sqrt{2}}$,      and      for
$\mu\le\frac{1}{\sqrt{2}}$  it is  fully local  (i.e., 3-local)  as in
this range the correlation does not violate any Bell inequality.

The canonical decomposition of the Svetlichny family (\ref{SvF})
is the ``noisy Svetlichny-box''
\begin{equation}
P^{\mu}_{SvF} =p_{Sv}P^{0000}_{Sv}+(1-p_{Sv})P_N; \quad 0<p_{Sv}\le 1, \label{nSv}
\end{equation}
with   $p_{Sv}=\mu/\sqrt{2}$,    which   is   a   special    case   of
Eq. (\ref{eq:psvl}).   Here, $P_N$ is  the maximally mixed  box, i.e.,
$P_N(abc|xyz)=1/8$  for all  $x,y,z,a,b,c$.  That  Eq. (\ref{nSv})  is
indeed canonical  decomposition is  proved in detail  in \cite{Jeb17}.
Briefly, the proof makes use of non-trivial symmetry properties of the
extreme boxes of the polytope  $\mathcal{R}$, such as for example that
every  Svetlichny  box  (say, $P_{Sv}^{0000}$)  has  a  ``complement''
(here:  $P_{Sv}^{0001}$) such  that their  uniform mixture  yields the
maximally                mixed               box,                i.e.,
$P_N=\frac{1}{2}(P^{0000}_{Sv}+P^{0001}_{Sv})$.     Note   that    the
Svetlichny family (\ref{SvF}) has  $\mathcal{G}=4\sqrt{2} \mu$ and the
Svetlichny-local box  $P_N$ in  the decomposition (\ref{nSv})  of this
family has $\mathcal{G}=0$,    indicating that this decomposition is
  indeed canonical. This implies that  the Svetlichny strength of the
Svetlichny family can be calculated from $\mathcal{G}$ and is given by
$\mathcal{G}(P^{\mu}_{SvF})/8=\frac{\mu}{\sqrt{2}}$ as stated above.

Therefore, the Svetlichny-box fraction in  Eq. (\ref{nSv}) can be read
off  as the  \textit{Svetlichny  strength} of  the Svetlichny  family.
Since the  Svetlichny family has  nonzero Svetlichny strength  for any
$\mu>0$,  the quantum  simulation  of these  correlations  by using  a
three-qubit  system necessarily  requires genuine  quantumness in  the
state \cite{Jeb17}.   In this light, the  Svetlichny strength $p_{Sv}$
satisfies the  relation $p_{Sv}=\frac{\sqrt{\tau_3}}{\sqrt{2}}$, where
$\tau_3=\sin^22\theta$  is  the  three-tangle  \cite{CKW00} of 
$\ket{\psi_{GGHZ}}$, when the Svetlichny family is simulated 
by $\ket{\psi_{GGHZ}}$ for  the
noncommuting projective  measurements corresponding to  the operators:
$A_0=\sigma_x$,   $A_1=\sigma_y$,  $B_0=(\sigma_x-\sigma_y)/\sqrt{2}$,
$B_1=(\sigma_x+\sigma_y)/\sqrt{2}$,         $C_0=\sigma_x$         and
$C_1=\sigma_y$.

As noted above, for $0  < \mu\le\frac{1}{\sqrt{2}}$ the fully local Svetlichny
family  can   be  decomposed  as   a  convex  mixture  of   the  3-local
deterministic boxes. In this range, the fully local Svetlichny family can be decomposed 
in the following $2$-local form across the bipartition $(A|BC)$:
\begin{align}
P^{\mu}_{SvF}&= \sum_{\lambda=0}^{3} p_{\lambda} P_{\lambda}^{Sv} (a|x) P_{\lambda}^{Sv} (b c|y z), \label{Svlhvm}
\end{align}
where $P_{\lambda} ^{Sv}(a|x)$ are different deterministic distributions and  $P_{\lambda}^{Sv} (b c|y z)$ are local boxes 
(see \ref{2lfsv} for the derivation of the above decomposition). 
For the fully  local  Svetlichny  family  ($0<  \mu  \leq
\frac{1}{\sqrt{2}}$),  the decomposition (\ref{Svlhvm}) defines a classical simulation 
protocol  
where Alice shares hidden variable $\lambda_{A-BC}$ of dimension $4$ with Bob-Charlie as in Fig. \ref{gsl}.

\begin{Theorem}
The  fully  local  Svetlichny family  $P^{\mu}_{SvF}$  ($0  <  \mu  \leq
\frac{1}{\sqrt{2}}$) is genuinely tripartite superlocal.
\label{thm:svet}
\end{Theorem}
\begin{proof} 
Let us try to reproduce the fully local  Svetlichny family
$P^{\mu}_{SvF}$ ($0 < \mu \leq \frac{1}{\sqrt{2}}$) in the scenario as
in  Fig.    \ref{gsl}  where  Alice  preshares   the  hidden  variable
$\lambda_{A-BC}$ of dimension $2$ with Bob-Charlie. Before proceeding,
we want  to mention that in  case of  Svetlichny  family, all the
marginal  probability  distributions of  Alice,  Bob  and Charlie  are
maximally mixed:
\begin{equation}
\label{mar}
P_{SvF}(a|x)  =  P_{SvF}(b|y)  = P_{SvF}(c|z)  =  \frac{1}{2}  \forall
a,b,c,x,y,z.
\end{equation}   
Let us now try  to  check whether  the  fully local  
Svetlichny family $P^{\mu}_{SvF}$ ($0  < \mu \leq \frac{1}{\sqrt{2}}$)
can be decomposed in the following form:
\begin{equation} \label{model2dim}
P^{\mu}_{SvF} = \sum_{\lambda=0}^{1} p_{\lambda} P_{\lambda}^{Sv} (a|x) P_{\lambda} ^{Sv}(b c|y z),
\end{equation}
where $p_0=x_0$, $p_1=x_1$  ($0 <x_0<1$, $0  <x_1<1$, $x_0+x_1  =1$)
and  $P_{\lambda}^{Sv} (b c|y z)$ are  local boxes. Let us  assume that Alice's
strategy   to be  deterministic   one, i.e.,  each   of  the   two  probability
distributions $P_{0}^{Sv} (a|x)$ and  $P_{1}^{Sv} (a|x)$ in the above decomposition
belongs to any one among $P_D^{00}$, $P_D^{01}$, $P_D^{10}$ and $P_D^{11}$. In
order to satisfy the  marginal probabilities for Alice $P_{SvF}(a|x)$,
the only  two possible choices  of $P_{0} ^{Sv}(a|x)$  and $P_{1}^{Sv}
(a|x)$ are:
\begin{enumerate}
\item $P_D^{00}$ and $P_D^{01}$ with $x_0=x_1=\frac{1}{2}$
\item $P_D^{10}$ and $P_D^{11}$ with $x_0=x_1=\frac{1}{2}$.
\end{enumerate}

Now, it can be easily checked  that none of these two possible choices
will  satisfy  all  the  tripartite  joint  probability  distributions
$P_{SvF}^{\mu}$  simultaneously (for  detailed  calculations, see   \ref{SvF2lhv}). 
It  is, therefore,  impossible to  reproduce the  fully local box
$P^{\mu}_{SvF}$ ($0 < \mu \leq \frac{1}{\sqrt{2}}$) in the scenario as
in  Fig.    \ref{gsl}  where  Alice  preshares   the  hidden  variable
$\lambda_{A-BC}$ of dimension $2$ with Bob-Charlie and uses deterministic 
strategy for each $\lambda_{A-BC}$.

The  fully local box
$P^{\mu}_{SvF}$ ($0 < \mu \leq \frac{1}{\sqrt{2}}$) cannot be reproduced
by a classical simulation model as in Eq. (\ref{model2dim}) with hidden  variable
$\lambda_{A-BC}$ of dimension $2$
even if Alice uses \textit{nondeterministic strategy} for each $\lambda_{A-BC}$.
To see this, we note that from any
decomposition of  the fully local box
in terms of fully deterministic boxes (\ref{DB}), one may derive a
classical simulation protocol as 
in  Fig.    \ref{gsl} with different deterministic distributions at Alice's side. Any such classical simulation protocol
does not require Alice to preshare the hidden variable $\lambda_{A-BC}$
of dimension more than $4$ since there are only $4$ possible different deterministic distributions given
by Eq. (\ref{alicedet}) at Alice's side. Hence, a classical simulation model with hidden  variable
$\lambda_{A-BC}$ of dimension $2$ of the  fully local box
$P^{\mu}_{SvF}$ ($0 < \mu \leq \frac{1}{\sqrt{2}}$) can be achieved by constructing a classical simulation model of the  fully local box
$P^{\mu}_{SvF}$ with hidden variable
$\lambda_{A-BC}$ of dimension $3$ or $4$ with different deterministic distributions at Alice's side 
followed by taking equal joint probability distributions at Bob-Charlie's side as common and making the 
corresponding distributions at Alice's side nondeterministic.

Let us now try check whether the fully local  noisy Svetlichny-box
$P^{\mu}_{SvF}$ ($0 < \mu  \leq \frac{1}{\sqrt{2}}$) can be simulated by a classical simulation 
model in the scenario as in Fig. \ref{gsl} where Alice shares the hidden variable of dimension $d_{\lambda_{A-BC}}=3$
and uses different deterministic strategy at each $\lambda_{A-BC}$. 
In this case, we assume that the box can be decomposed in the following way:
\begin{equation}
P^{\mu}_{SvF} = \sum_{\lambda=0}^{2} p_{\lambda} P_{\lambda} ^{Sv}(a|x) P_{\lambda}^{Sv} (b c|y z).
\end{equation}
Here, $p_0= x_0$,  $p_1 = x_1$, $p_2  = x_2$ ($0 <x_0<1$, $0  <x_1<1$, $0 <x_2<1$,
$x_0+x_1+x_2 =1$) and $P_{\lambda} ^{Sv}(a|x)$ are deterministic distributions and
$P_{\lambda}^{Sv} (b c|y z)$ are local boxes. 
Since Alice's distributions are  deterministic, the three probability distributions 
$P_{0} ^{Sv}(a|x)$, $P_{1} ^{Sv}(a|x)$
and  $P_{2} ^{Sv}(a|x)$  must be  equal to  any three among $P_D^{00}$,
$P_D^{01}$, $P_D^{10}$  and $P_D^{11}$. But any  such combination will
not satisfy  the marginal  probabilities $P_{SvF}(a|x)$ for  Alice. This implies that the  fully local box
$P^{\mu}_{SvF}$ ($0 < \mu \leq \frac{1}{\sqrt{2}}$) cannot be reproduced in
any classical simulation protocol with different deterministic distributions $P_{\lambda} ^{Sv}(a|x)$ at Alice's side, 
where   Alice  preshares  the   hidden  variable
$\lambda_{A-BC}$ of dimension $3$ with Bob-Charlie.

Therefore, in the  classical simulation model for the fully local Svetlichny family in the scenario 
as in Fig. \ref{gsl} where Alice uses deterministic strategies, Alice has to share the hidden variable of dimension  
$d_{\lambda_{A-BC}}=4$.  

Suppose the fully local  Svetlichny family
$P^{\mu}_{SvF}$ ($0 < \mu \leq \frac{1}{\sqrt{2}}$) can be reproduced by the 
following classical simulation model:
\begin{equation}
\label{new1}
P^{\mu}_{SvF} = \sum_{\lambda=0}^{3} p_{\lambda} P_{\lambda}^{Sv} (a|x) P_{\lambda}^{Sv} (b c|y z),
\end{equation} 
where $P_{\lambda}^{Sv} (a|x)$ are different deterministic distributions
and either any three of the four joint probability distributions $P_{\lambda}^{Sv} (b c|y z)$ 
are equal to each other, or there exists two sets each containing two equal joint probability 
distributions $P_{\lambda}^{Sv} (b c|y z)$; $0 < p_{\lambda} < 1$ for $\lambda$ = $0,1,2,3$; 
$\sum_{\lambda=0}^{3} p_{\lambda} = 1$. Then taking equal joint probability distributions 
$P_{\lambda}^{Sv} (b c|y z)$ at Bob-Charlie's side as common and making corresponding distribution 
at Alice's side non-deterministic will reduce the dimension of the hidden variable $\lambda_{A-BC}$ 
from $4$ to $2$. For example, let us consider
\be
P_{0}^{Sv} (b c|y z) = P_{1}^{Sv} (b c|y z) = P_{2}^{Sv} (b c|y z).
\ee
Now in order to satisfy Alice's marginal given by Eq. (\ref{mar}), one must take $p_0$ = $p_1$ = $p_2$ = $p_3$ = $\frac{1}{4}$. 
Hence, the decomposition (\ref{new1}) can be written as, 
\begin{equation}
\label{new2}
P^{\mu}_{SvF} = q_0 \mathbb{P}^{Sv}_{0}(a|x) P_{0}^{Sv} (b c|y z) + p_{3} P_{3}^{Sv} (a|x) P_{3}^{Sv} (b c|y z),
\end{equation} 
where  $
\mathbb{P}^{Sv}_{0}(a|x) = \frac{P_{0}^{Sv} (a|x)+  P_{1}^{Sv} (a|x)+ P_{2}^{Sv} (a|x)}{3},
       $
which is a non-deterministic distribution at Alice's side, and
$
q_0 = \frac{3}{4}.
$
The decomposition (\ref{new2}) represents a classical simulation 
protocol of the fully local  Svetlichny family $P^{\mu}_{SvF}$ ($0 < \mu \leq \frac{1}{\sqrt{2}}$) 
with different deterministic/non-deterministic distributions at Alice's side, where Alice shares hidden variable
$\lambda_{A-BC}$ of dimension $2$ with Bob-Charlie. Now in this protocol, considering arbitrary joint probability 
distributions $P_{\lambda}^{Sv} (b c|y z)$ at Bob-Charlie's side (without considering any constraint as 
in the case presented in \ref{SvF2lhv}),
it can be checked
that all the tripartite distributions of $P^{\mu}_{SvF}$ are not reproduced simultaneously.

There are the following other cases in which the dimension of the hidden variable $\lambda_{A-BC}$ 
can be reduced from $4$ to $2$ in the classical simulation model as in Eq. (\ref{new1}): 
\begin{center}
$P_{0}^{Sv} (b c|y z) = P_{2}^{Sv} (b c|y z) = P_{3}^{Sv} (b c|y z)$;  \\ 
$P_{0}^{Sv} (b c|y z) = P_{1}^{Sv} (b c|y z) = P_{3}^{Sv} (b c|y z)$;   \\
$P_{1}^{Sv} (b c|y z) = P_{2}^{Sv} (b c|y z) = P_{3}^{Sv} (b c|y z)$;  \\
$P_{0}^{Sv} (b c|y z)$ = $P_{1}^{Sv} (b c|y z)$ as well as $P_{2}^{Sv} (b c|y z)$ = $P_{3}^{Sv} (b c|y z)$;\\ 
$P_{0}^{Sv} (b c|y z)$ = $P_{2}^{Sv} (b c|y z)$ as well as $P_{1}^{Sv} (b c|y z)$ = $P_{3}^{Sv} (b c|y z)$;\\ 
$P_{0}^{Sv} (b c|y z)$ = $P_{3}^{Sv} (b c|y z)$ as well as $P_{1}^{Sv} (b c|y z)$ = $P_{2}^{Sv} (b c|y z)$. \\
\end{center}
Now in any of these possible cases, considering arbitrary joint probability distributions 
$P_{\lambda}^{Sv} (b c|y z)$ at Bob-Charlie's side (without considering any constraint), 
it can be checked that all the tripartite distribution $P^{\mu}_{SvF}$ are not reproduced simultaneously. 
Hence, this also holds when the boxes $P_{\lambda}^{Sv} (b c|y z)$ satisfy 
NS principle as well as locality condition.

Hence, one can conclude that it is impossible to reduce the dimension from $4$ to $2$ in the 
classical simulation protocol of the fully local  Svetlichny family $P^{\mu}_{SvF}$ ($0 < \mu \leq \frac{1}{\sqrt{2}}$)
in  the scenario  as  in Fig.   \ref{gsl} .

It  is, therefore,  impossible to  reproduce the  fully local box
$P^{\mu}_{SvF}$ ($0 < \mu \leq \frac{1}{\sqrt{2}}$) with deterministic/non-deterministic 
distributions at Alice's side,  where  Alice  preshares   the  hidden  variable
$\lambda_{A-BC}$ of dimension $2$ with Bob-Charlie  in the scenario as in  Fig.    \ref{gsl}.

It can be checked that the  fully local Svetlichny box $P^{\mu}_{SvF}$ ($0 <
\mu  \leq  \frac{1}{\sqrt{2}}$)  is  non-product.  It  is,  therefore,
impossible to reproduce the  fully local box $P^{\mu}_{SvF}$ ($0 < \mu
\leq \frac{1}{\sqrt{2}}$)  in the  scenario where Alice  preshares the
hidden variable $\lambda_{A-BC}$ of dimension $1$ with Bob-Charlie.

Hence, the dimension of the hidden variable $\lambda_{A-BC}$, 
which Alice preshares with Bob-Charlie to reproduce the fully local box 
$P^{\mu}_{SvF}$ ($0 < \mu \leq \frac{1}{\sqrt{2}}$), must be \textit{greater than $2$}.
Therefore, 
the fully local Svetlichny box $P^{\mu}_{SvF}$ ($0 < \mu \leq \frac{1}{\sqrt{2}}$) has 2-local form across the bipartite cut $(A|BC)$ and is not $2$-sublocal across that bipartite cut. Hence, the fully local Svetlichny box $P^{\mu}_{SvF}$ ($0 < \mu \leq \frac{1}{\sqrt{2}}$) 
is superlocal across that bipartite cut.

It can be  checked that a similar argument holds  across the remaining
two bipartite  cuts $(C|AB)$  and $(B|CA)$, 
i.e., this box is superlocal across all three possible bipartite cuts
and, hence, must be genuinely tripartite superlocal.
\end{proof}

Since by definition, fully local correlations, which are genuinely superlocal, must be
absolutely superlocal as well, it follows from Theorem \ref{thm:svet} that the fully
 local  Svetlichny   box   $P^{\mu}_{SvF}$  ($0   <  \mu   \leq
\frac{1}{\sqrt{2}}$) is absolutely superlocal also.

\subsection{Noisy Mermin-box \label{sec:mermin}}

We  are now  interested in  quantum  correlations that  belong to  the 
Mermin  family  defined  as 
\begin{equation}  P^{\nu}_{MF}=\frac{1+(-1)^{a\oplus  b
    \oplus  c \oplus  xy\oplus xz  \oplus yz}\delta_{x\oplus  y \oplus
    1,z}\nu}{8};    \quad   0<\nu\le    1.    \label{MeF}   
\end{equation}
For $\nu\le\frac{1}{2}$, the above box is fully local as in this range
the correlation does  not violate any Bell inequality.   The above box
is  absolutely  nonlocal,  but  2-local for  $\nu>\frac{1}{2}$  as  it
violates  the  Mermin  inequality   (given  in  Eq.   (\ref{MI}))  for
$\nu>\frac{1}{2}$, but  not any of the  Svetlichny inequalities. 
  Thus,  it isn't  obvious that  this correlation  would be  genuinely
  superlocal, yet this is what will be established below.

The Mermin family (\ref{MeF}) has the canonical decomposition
as the noisy Mermin-box
\begin{equation}  
P^{\nu}_{MF}=p_{M}P_{M}+(1-p_{M})P_N; \quad 0<p_{M}\le  1, \label{nMe}
\end{equation}      
with         $p_{M}=\nu$,  which is a special case of Eq. (\ref{3dfact1}). 
Here,        the         Mermin-box
$P_M=\frac{1}{2}\left(P^{0000}_{Sv}+P^{1110}_{Sv}\right)$.   That  Eq.
(\ref{nMe}) is indeed the  canonical decomposition for Eq. (\ref{MeF})
can be  shown, as with  the case of the  noisy Svetlichny box, 
by making  use of  the non-trivial  symmetry properties  of the
extremal boxes of the  polytope $\mathcal{R}$, such as the
fact that any given Mermin box (say, $P_M$) has a complement such that
their    uniform   mixture    yields    the    white   noise,    i.e.,
$P_N=\frac{1}{2}P_M+\frac{1}{2}P^\prime_M$, where  $P'_M = \frac{1}{2}
\left(P^{0001}_{Sv}+P^{1111}_{Sv}\right)$   \cite{Jeb17}.   Therefore,
the Mermin-box fraction in Eq. (\ref{nMe}) indeed gives Mermin strength
\cite{Jeb17} of the Mermin family.  Since the Mermin family has 
a decomposition as in Eq. (\ref{3dfact1}), its Mermin strength can be
calculated from $\mathcal{Q}$ and is given by $\mathcal{Q}(P^{\nu}_{MF})/4=\nu$ as 
stated above.

For  any
$\nu>0$,  the quantum  simulation  of  the Mermin  family  by using  a
three-qubit  system necessarily  requires genuine  quantumness in  the
state, even if it is fully local or $2$-local.  This is due to the
fact  that the  Mermin  family  has nonzero  Mermin  strength for  any
$\nu>0$ \cite{Jeb17}.   Note that the GGHZ  state, $\ket{\psi_{GGHZ}}$
(given  by Eq. (\ref{gghz})),  gives  rise to  the  Mermin family  with
Mermin strength  $\nu= \sqrt{\tau_3}$ for the  noncommuting projective
measurements   corresponding  to   the  operators:   $A_0  =\sigma_x$,
$A_1=\sigma_y$,  $B_0=\sigma_x$, $B_1=\sigma_y$,  $C_0 =\sigma_x$  and
$C_1 = \sigma_y$ that demonstrates the GHZ paradox.

As noted above, for $0 < \nu\le\frac{1}{2}$ the fully local noisy Mermin box can be
decomposed in a convex mixture of the 3-local deterministic boxes.
In this range, the fully local Mermin family can be decomposed in the following $2$-local 
form across the bipartition $(A|BC)$:

\begin{align} 
P^{\nu}_{MF}&= \sum_{\lambda=0}^{3} r_{\lambda} P_{\lambda}^{M} (a|x) P_{\lambda} ^{M}(b c|y z), \label{mlhvm}
\end{align}
where $P_{\lambda} ^{M}(a|x)$ are different deterministic distributions and $P_{\lambda}^{M} (b c|y z)$ are local boxes (see \ref{2lfm} for the derivation of the above decomposition). 
For the fully  local  noisy  Mermin-box  ($0< \nu \leq \frac{1}{2}$),  the decomposition (\ref{mlhvm}) defines a classical simulation 
protocol
where Alice shares hidden variable $\lambda_{A-BC}$ of dimension $4$ as in Fig. \ref{gsl}. 

\begin{Theorem} 
The  fully  local noisy  Mermin  box  $P^{\nu}_{MF}$  ($0 <  \nu  \leq
\frac{1}{2}$) is genuinely tripartite superlocal.
\label{thm:mermin}
\end{Theorem}
\begin{proof} 
In  case of noisy Mermin-box also,  all  the marginal  probability distributions  of
Alice, Bob and Charlie are maximally mixed:
\begin{equation}
P_{MF}(a|x)  =   P_{MF}(b|y)  =  P_{MF}(c|z)  =   \frac{1}{2}  \forall
a,b,c,x,y,z.
\end{equation} 

Let us try to construct a classical simulation protocol for the fully local 
Mermin family $P^{\nu}_{MF}$ ($0< \nu \leq \frac{1}{2}$) with different deterministic distributions $P_{\lambda}^{M}(a|x)$ at Alice's side, 
where Alice shares hidden variable $\lambda_{A-BC}$ of dimension $2$ with Bob-Charlie. In  this case,  the
fully  local   Mermin  family ($0<  \nu  \leq  \frac{1}{2}$) can be
decomposed in the following way:
\begin{equation}
P^{\nu}_{MF} = \sum_{\lambda=0}^{1} r_{\lambda} P_{\lambda} ^{M}(a|x) P_{\lambda}^{M} (b c|y z).
\end{equation}
Here, $r_0= a_0$, $r_1 = a_1$ ($0 <a_0<1$, $0 <a_1<1$, $a_0 + a_1 =1$). Since Alice's distributions 
are deterministic, the two probability distributions $P_{0}^{M} (a|x)$ and 
$P_{1}^{M} (a|x)$ must be equal to any two among $P_D^{00}$, $P_D^{01}$, $P_D^{10}$ and $P_D^{11}$. 
In order to satisfy the marginal probabilities for Alice $P_{MF}(a|x)$, the only two possible choices 
of $P_{0}^{M} (a|x)$ and $P_{1} ^{M}(a|x)$ are:\\
1) $P_D^{00}$ and $P_D^{01}$ with $a_0=a_1=\frac{1}{2}$\\
2) $P_D^{10}$ and $P_D^{11}$ with $a_0=a_1=\frac{1}{2}$.

Now, in a similar process as adopted in case of Svetlichny family, it can be easily checked 
that none of these two possible choices will satisfy all the tripartite joint probability distributions 
$P^{\nu}_{MF}$ simultaneously.

Since Alice's marginal distributions are maximally mixed, it can be shown in a similar way as presented in case of 
noisy Svetlichny box that there does not exist a classical simulation model for the fully local Mermin family in the scenario as in Fig. \ref{gsl} 
where Alice shares the hidden variable of dimension $d_{\lambda_{A-BC}}=3$ and uses different deterministic strategies at each $\lambda_{A-BC}$.

Thus, in the  classical simulation model for the fully local Mermin family in the scenario 
as in Fig. \ref{gsl} where Alice uses different deterministic strategies, Alice has to share the hidden variable of dimension  
$d_{\lambda_{A-BC}}=4$. 
Let us now try to check whether there exists such a classical simulation model 
for  the  fully  local noisy  Mermin  box where $d_{\lambda_{A-BC}}$ can be reduced from $4$ to $2$ by allowing non-deterministic strategies on Alice's side. That is
we try to construct the following classical simulation protocol for the fully  local  noisy  
Mermin-box  ($0< \nu \leq \frac{1}{2}$) with different deterministic distributions $P_{\lambda} ^{M}(a|x)$ at Alice's side,
where Alice shares hidden variable $\lambda_{A-BC}$ of dimension $4$ with Bob-Charlie in  the scenario  as  in Fig. \ref{gsl}:
\begin{equation}
P^{\nu}_{MF} = \sum_{\lambda=0}^{3} p_{\lambda} P_{\lambda}^{M} (a|x) P_{\lambda}^{M} (b c|y z),
\end{equation} 
where either any three of the four joint probability distributions $P_{\lambda}^{SM} (b c|y z)$ are equal to each other, 
or there exists two sets each containing two equal joint probability distributions 
$P_{\lambda}^{M} (b c|y z)$; $0 < p_{\lambda} < 1$ for $\lambda$ = $0,1,2,3$; $\sum_{\lambda=0}^{3} p_{\lambda} = 1$. 
Then, as described earlier in the case of noisy
Svetlichny-box, taking equal joint probability distributions $P_{\lambda}^{M} (b c|y z)$ at Bob-Charlie's side 
as common and making corresponding distribution at Alice's side non-deterministic will reduce the dimension
of the hidden variable $\lambda_{A-BC}$ from $4$ to $2$. Now following the similar procedure adopted in case of noisy Svetlichny box,
one can show that it is impossible to reduce the dimension from $4$ to $2$ in the classical simulation protocol of 
the fully  local  noisy   Mermin-box  ($0< \nu \leq \frac{1}{2}$) in  the scenario  as  in Fig.   \ref{gsl}.

It is, therefore, impossible to reproduce the fully 
local  Mermin family $P^{\nu}_{MF}$ ($0< \nu \leq \frac{1}{2}$) in any classical simulation protocol 
with deterministic/non-deterministic distributions $P_{\lambda} ^{M}(a|x)$ at Alice's side, 
where Alice preshare the hidden variable $\lambda_{A-BC}$ of dimension $2$ with Bob-Charlie in the scenario as 
in Fig. \ref{gsl} .

 It can be checked that the fully local  Mermin family 
 $P^{\nu}_{MF}$ ($0< \nu \leq \frac{1}{2}$) is non-product. It is, therefore, impossible to reproduce 
 the fully local  Mermin family $P^{\nu}_{MF}$ ($0< \nu \leq \frac{1}{2}$) in the scenario where
 Alice preshares the hidden variable $\lambda_{A-BC}$ of dimension $1$ with Bob-Charlie.

 Hence, the dimension of the hidden variable $\lambda_{A-BC}$, which Alice preshares with 
Bob-Charlie to reproduce the fully local  Mermin family 
$P^{\nu}_{MF}$ ($0< \nu \leq \frac{1}{2}$), must be \textit{greater than $2$}. 
Therefore, 
the fully local  Mermin family $P^{\nu}_{MF}$ ($0< \nu \leq \frac{1}{2}$) has 2-local form across the bipartite cut $(A|BC)$ and is not 2-sublocal across that bipartite cut. Hence, the fully local  Mermin family 
$P^{\nu}_{MF}$ ($0< \nu \leq \frac{1}{2}$) is superlocal across the bipartite cut $(A|BC)$.

It can be  checked that a similar argument holds  across the remaining
two bipartite  cuts $(C|AB)$  and $(B|CA)$, 
i.e., this box is superlocal across all three possible bipartite cuts
and, hence, must be genuinely tripartite superlocal.
\end{proof}

 Since the  fully  local noisy  Mermin  box  $P^{\nu}_{MF}$  ($0 <  \nu  \leq
\frac{1}{2}$) is genuinely superlocal, it is absolutely superlocal as well.

Now,  as noted  before, for  $0 <  \nu\le 1$  the 2-local 
Mermin  family can  be decomposed  in a  convex mixture  of the  local
vertices and 2-local vertices. In this range, the Mermin family can be decomposed in the following $2$-local 
form across the bipartition $(A|BC)$: 
\begin{align}
P^{0<\nu<1}_{MF}&= \sum_{\lambda=0}^{3} r_{\lambda} P_{\lambda}^{M} (a|x) P_{\lambda} ^{M}(b c|y z), \label{mlhv3m}
\end{align}
where $P_{\lambda} ^{M}(a|x)$ are different deterministic distributions and $P_{\lambda}^{M} (b c|y z)$ are NS boxes 
(see \ref{2lfm2} for the derivation of the above decomposition). 
For the   $2$-local  noisy   Mermin-box  ($0<  \nu  \leq
1$),  the decomposition (\ref{mlhv3m}) defines a classical simulation 
protocol as in Fig. \ref{sl}   where Alice shares hidden variable
$\lambda_{A-BC}$ of dimension $4$.

Since the fully local noisy Mermin box is  a special case
of the 2-local noisy Mermin box and the proof of Theorem \ref{thm:mermin} is 
independent of the locality condition of the bipartite distributions $P_{\lambda}^{M} (b c|y z)$ at Bob-Charlie's side, 
the proof is also valid when the bipartite distributions $P_{\lambda}^{M} (b c|y z)$ 
at Bob-Charlie's side are NS (local or nonlocal) boxes. Hence, it is not difficult to see that
the proof  of the Theorem \ref{thm:mermin} can  be straightforwardly  adopted to
obtain the following result:
\begin{Theorem}
 The 2-local  noisy Mermin box $P^{\nu}_{MF}$ ($0 <  \nu \leq 1$)
is genuinely tripartite superlocal.
\end{Theorem}


\section{Connection between genuine super-locality and genuine nonclassicality}\label{gslgqd}

In Ref.  \cite{GTC}, Giorgi  \etal defined genuine  tripartite quantum
discord   to  quantify   the  quantum   part  of   genuine  tripartite
correlations in  a tripartite  quantum state.   As the  name suggests,
this kind of quantifier also captures genuine quantumness of separable
states.     Genuine   tripartite    quantum    discord   defined    in
Ref. \cite{GCTS} goes to zero  \emph{iff} there exists a bipartite cut
of  the  tripartite system  such  that  no quantum  correlation  exists
between the two parts.  It is known that a bipartite quantum state has
no quantum correlation  as quantified by Alice to  Bob quantum discord
\emph{iff}  it can  be  written in  the  classical-quantum state  form
\cite{Dakicetal}.   We define  tripartite classical-quantum  states as
follows.
\begin{definition}
A fully separable tripartite state  has a classical-quantum state form
across  the bipartite  cut  $(A|BC)$  if it  can  be  decomposed as  
\begin{equation}
\rho^{A|BC}_{CQ}=\sum_ip_i |i\rangle^{A}\langle  i| \otimes \rho^{B}_i
\otimes  \rho^{C}_i,  \label{cqBC} 
\end{equation}  
where  $\{|i\rangle^{A}\}$  is some  orthonormal  basis  of
Alice's Hilbert space $\mathcal{H}_A$.
\end{definition}
Note that  the classical-quantum  states as  defined above  do not have
nonzero  genuine quantum  discord  since Alice's  subsystem is  always
classically  correlated   with  Bob   and  Charlie's   subsystem.   We
characterize  a   (the  fully  separable)  state   as  having  genuine
quantumness, if it cannot be  written in the classical-quantum state in
any bipartite cut as in Eq.  (\ref{cqBC}).

  The  Svetlichny family and  Mermin family violate  a three-qubit
  biseparability    inequality    for    $\mu>1/2$    and    $\nu>1/2$
  \cite{DBM+17}.  Therefore,  in that range these  two families
  certify      genuine     three-qubit      entanglement     in      a
  semi-device-independent   way  as   well  \cite{DBM+17,JEB16}.   For
  $\mu\le1/2$ and $\nu\le1/2$, the  Svetlichny and Mermin families can
  also be reproduced by separable three-qubit states since they do not
  violate any biseparability inequality in this range.

However, for  $\mu, \nu  \in (0,\frac{1}{2}]$,  the simulation  of the
  Svetlichny family  and Mermin family by  using three-qubit separable
  states serves to witness genuine  quantumness in the form of genuine
  quantum discord as they have nonzero Svetlichny strength and nonzero
  Mermin strength, respectively \cite{Jeb17}. This observation prompts
  us to make the following observation.
\begin{observation}
Genuine quantumness (i.e., nonzero discord across any bipartite cut) of 
any correlation  $P$  is  necessary  for  genuine superlocality.
\end{observation}
\begin{proof}
  Consider  tripartite   boxes  arising   from  three-qubit
  classical-quantum   states  which   have  the   form  as   given  in
  Eq. (\ref{cqBC})  with $i=0,1$.  In  the Bell scenario that  we have
  considered,  for Alice  measuring  in basis  $\{|i\rangle\}$, it  is
  clear that  the resulting box  can be simulated by  a probabilistic
  strategy using  dimension $d_{\lambda_A}=2$  on Alice's  side.  This
  observation  holds  even when  Alice  measures  in any  other  basis
  (except that  her random number  generator will be possibly  be more
  randomized). This  implies that for  any three-qubit state  which do
  not have genuine quantumness, there  exists a bipartite cut in which
  it  is  not  superlocal.   Therefore,  genuine  quantum  discord  is
  necessary for implying genuine superlocality.  
\end{proof}

\section{Genuine multipartite nonclassicality}\label{gmnl}

Generalizing  the  definitions  presented  in  Section  \ref{prl},  an
$n$-partite  correlation  is  said   to  be  \textit{fully}  local  or
$n$-local if the box has a decomposition of the form:
\begin{align}  
P(a_1,a_2,&\cdots,a_n|x_1,x_2,\cdots,x_n)= \nonumber \\
    &\sum_{\lambda}
p_\lambda  P_\lambda(a_1|x_1)P_\lambda(a_2|x_2)
\cdots P_\lambda(a_n|x_n), 
\label{lhvn} 
\end{align} 
where $\sum_{\lambda} p_\lambda=1$, and $x_1, x_2, \cdots, x_n$ denote
the inputs  (measurement choices) and  $a_1, a_2, \cdots,  a_n$ denote
the outputs (measurement  outcomes) of the parties  $q_1$, $q_2$, ...,
$q_n$ respectively.  An $n$-partite box  that is not $n$-local is said
to have ``absolute nonlocality''.

A  $n$-partite correlation  is  said  to be  $k$-local  if  it can  be
decomposed as a  convex combination of $k$-partitions  such that these
$k$ parts (defined by each  $k$-partition) are locally correlated with
each other. For example, a  $4$-partite correlation is $3$-local if it
can be  decomposed as  a convex combination  of tripartitions  such as
$(q_1|q_2 q_3|q_4)$ such that these  parts are locally correlated each
other.  However,  $q_2 q_3$ may  be nonlocal in itself.   Similarly, a
$4$-partite  correlation is  $2$-local if  it can  be decomposed  as a
convex combination of probability distributions over bipartitions such
that in each bipartition, the two parts are locally correlated, though
within a  part, even nonlocality  may hold. Any $k$-local  correlation is
also $k'$-local where $k^\prime < k$.  Thus, a $4$-local correlation is also
$3$-local  and a  $3$-local  correlation is  also  $2$-local. But  the
converse  is not  true. 

Therefore,  the weakest  form  of locality  is  $2$-locality, and  the
strongest form of nonlocality for  an $n$-partite system is that which
is not $2$-local.  This is called genuine $n$-partite nonlocality, for
which all bipartitions are nonlocal.

An
$n$-partite system  is $n$-sublocal (or, fully sublocal) if each of the $n$ particles  are locally
correlated,  with  the  shared classical  randomness dimension being  less than or equal to the
smallest local Hilbert space dimension among all the  $n$  particles, in analogy with Eq. (\ref{eq:3floc}):
\begin{definition}
Suppose we have an  $n$-partite quantum state in $\mathbb{C}^{d_{s_1}}
\otimes      \mathbb{C}^{d_{s_2}}      \otimes     \cdots      \otimes
\mathbb{C}^{d_{s_n}}$   and  measurements   which   produce  a  fully  local
$n$-partite   box   $P(a_1a_2\cdots   a_n|x_1x_2\cdots   x_n)$.    The
correlation $P$ is \textit{$n$-sublocal} (or, \textit{fully sublocal})
precisely if there exists a decomposition such that
\begin{align}  
P(a_1a_2&\cdots a_n|x_1x_2\cdots x_n)= \nonumber \\
    &\sum^{d_\lambda-1}_{\lambda=0}
p_\lambda  P_\lambda(a_1|x_1)P_\lambda(a_2|x_2)
\cdots P_\lambda(a_n|x_n), 
\label{eq:nfloc} 
\end{align} 
with   $d_\lambda\le   d$,   and   $d=\min\{d_{s_1},d_{s_2},   \cdots,
d_{s_n}\}$, where  $d_{s_j}$ is the local Hilbert space dimension of the  $j$th particle.
The fully local correlation  $P$   is   absolutely  superlocal   if  it   is not
$n$-sublocal. In other words, absolute $n$-partite superlocality holds
iff  there is  no $n$-sublocal  decomposition (\ref{eq:nfloc})  of the
given  fully local box.  Here
$\sum^{d_\lambda-1}_{\lambda=0}   p_\lambda  =1$.\hfill $\blacksquare$
\end{definition}

A $n$-partite  correlation $P$ is  said to be ``$k$-sublocal  across a
particular  $k$-partition''   if  these   $k$  parts   are  sublocally
correlated with each other.  For example, a $4$-partite correlation is
$3$-sublocal across  the tripartite  cut $(q_1|q_2 q_3|q_4)$  if these
three parts are sublocally correlated  with each other.  Note that the
part $(q_2 q_3)$ may be sublocal or superlocal or even nonlocal.

Thus, a  correlation which  is $4$-sublocal across  some $4$-partition
$Q^{(|||)}$, is  also $3$-sublocal across any  tripartition $Q^{(||)}$
obtained  by  merging any  two  partitions  of $Q^{(|||)}$,  and  also
$2$-sublocal across any bipartition  $Q^{(|)}$ obtained by merging any
two partitions of $Q^{(||)}$.

In  general,  any  correlation,  which  is  $k$-sublocal  across  some
$k$-partition,    is    also     $k^\prime$-sublocal    across    some
$k^\prime$-partition    where    $k^\prime    <    k$,    where    the
$k^\prime$-partition  has  been  obtained  by  merging  parts  of  the
$k$-partition.   Therefore,   the  weakest  form  of   sublocality  is
$2$-sublocality.   \textit{An $n$-partite  correlation $P$  that isn't
  sublocal   --  or   equivalently,  is   superlocal--  across   every
  bipartition, is genuinely $n$-partite superlocal.}        Note that,
  because of the non-convexity associated with sublocal sets, a convex
  combination   of   2-sublocal   correlations   needn't   itself   be
  2-sublocal.

In line with our  definition for superlocality for a multipartite
system, we may  define \textit{$n$-concord} as the  absence of quantum
discord   across  all   cuts  splitting   all  $n$   particles,  i.e.,
$(q_1|q_2|\cdots|q_n)$.  For example, the  system $q_1 q_2 q_3 q_4$ has 4-concord
if there  is no  quantum discord  across each  of the  cuts $q_1|q_2|q_3|q_4$.
Absolute $n$-partite  discord \cite{MPS+10} holds when  the state   in  
question is not $n$-concordant. Genuine $n$-partite discord holds  when  
the state in question  lacks  $2$-concord form  across all possible bipartitions 
\cite{GTC,MZ12,GCTS,GTQD}. The   
relation between  different such measures of discord and superlocality 
for multipartite systems, such as that noted for the        tripartite       
system in Section \ref{gslgqd}, is an interesting issue       meriting 
further studies. 

\section{Conclusions}\label{conc}

In Ref.  \cite{Jeb17}, two  quantities called, Svetlichny strength and
Mermin strength, have been introduced to study genuine nonclassicality
of tripartite  correlations.  By  using these  two quantities,  it has
been demonstrated that genuine tripartite quantum discord is necessary
to simulate certain fully local  or 2-local tripartite correlations if
the  measured  tripartite systems  are  restricted  to be  three-qubit
states.

Genuine multipartite quantum nonlocality occurs if the multipartite correlation 
cannot be expressed as a convex combination of all possible bipartitions where the $2$-parts 
defined by each of these bipartitions are locally correlated with each other. 
Our motivation has been to
perform   the
characterization  of  genuine   nonclassicality  of  local  (fully or partially local) tripartite
correlations arising  from the concept of  superlocality, and relating
this to  genuine quantum discord.  Thus,  genuine superlocality, i.e.,
the occurrence  of superlocality across all  bipartitions, provides an
operational definition  of genuine  nonclassicality.  We  have studied
how genuine superlocality  occurs for two families of  local  tripartite 
correlations having their 
nonclassicality quantified in terms of nonzero Svetlichny strength
and nonzero Mermin strength, respectively. 

In Ref. \cite{BP14}, it was demonstrated that certain bipartite separable  states having  quantum discord may 
improve  the so-called random access codes (RAC), which is a
class of communication problem,  if the shared
randomness between the two parties is  limited to be finite.
Recently, in Ref. \cite{HSM+17}, a family of RAC protocols have been considered in tripartite 
quantum networks and is associated with genuine tripartite nonlocality. Our work on characterizing 
genuine quantumness of certain local tripartite or multipartite correlations in the limited dimensional simulation scenario
and its link with genuine quantum discord motivates the following study. 
It would be interesting to investigate  quantum advantage for the above RAC protocol associated 
with tripartite or multipartite quantum networks in the presence of limited 
shared randomness by using tripartite or multipartite separable states with genuine quantum discord.

\section*{Acknowledgements}

CJ thanks  Manik Banik and  Arup Roy for discussions.  DD acknowledges
the  financial  support  from   University  Grants  Commission  (UGC),
Government  of  India. CJ,  SG  and  ASM acknowledge  support  through
Project SR/S2/LOP-08/2013 of the DST, Govt.  of India.

\newpage 

\bibliographystyle{unsrt.bst}

\bibliography{JEB}

\newpage

\appendix

\section{$2$-local form across the bipartition $(A|BC)$ for the fully  local  Svetlichny family  $P^{\mu}_{SvF}$  in the range $0  <  \mu  \leq
\frac{1}{\sqrt{2}}$} \label{2lfsv}

For $0  < \mu \le \frac{1}{\sqrt{2}}$ the fully local Svetlichny
family  can   be  decomposed  as   a  convex  mixture  of   the  3-local
deterministic boxes. In this range, we consider the following decomposition for the 
Svetlichny family
in terms of the 3-local
deterministic boxes: 
\begin{align}
P^{\mu}_{SvF}&=\frac{1}{4} P_D^{00} \Bigg\{ \frac{\sqrt{2} \mu}{4} \big( P^{0010}_D+P^{0111}_D+P^{1000}_D+P^{1101}_D \big) \nonumber \\
&+ \frac{1- \sqrt{2} \mu}{4} \big( P_D^{1000} + P_D^{1100} + P_D^{1101} + P_D^{1001} \big) \Bigg\} \nonumber \\
& + \frac{1}{4} P_D^{01} \Bigg\{ \frac{\sqrt{2} \mu}{4} \big( P^{0011}_D+P^{0110}_D+P^{1001}_D+P^{1100}_D \big) \nonumber \\ 
&+ \frac{1- \sqrt{2} \mu}{4} \big( P_D^{1000} + P_D^{1100} + P_D^{1101} + P_D^{1001} \big) \Bigg\} \nonumber \\ 
& + \frac{1}{4} P_D^{10} \Bigg\{ \frac{\sqrt{2} \mu}{4} \big( P^{0000}_D+P^{0101}_D+P^{1011}_D+P^{1110}_D \big) \nonumber \\
&+ \frac{1- \sqrt{2} \mu}{4} \big( P_D^{0000} + P_D^{0101} + P_D^{0001} + P_D^{0100} \big) \Bigg\} \nonumber \\ 
& + \frac{1}{4} P_D^{10} \Bigg\{ \frac{\sqrt{2} \mu}{4} \big( P^{0001}_D+P^{0100}_D+P^{1010}_D+P^{1111}_D \big) \nonumber \\
&+ \frac{1- \sqrt{2} \mu}{4} \big( P_D^{0000} + P_D^{0101} + P_D^{0001} + P_D^{0100} \big) \Bigg\} \nonumber \\ 
&: = \sum_{\lambda=0}^{3} p_{\lambda} P_{\lambda}^{Sv} (a|x) P_{\lambda}^{Sv} (b c|y z), \label{Svlhv}
\end{align}
where $p_0$ = $p_1$ = $p_2$ = $p_3$ = $\frac{1}{4}$;
\begin{align}
&P_{0}^{Sv} (a|x) = P_D^{00}, \quad P_{1} ^{Sv}(a|x) = P_D^{01}, \nonumber \\
&P_{2}^{Sv}(a|x) = P_D^{10}, \quad P_{3}^{Sv}(a|x) = P_D^{11}; \nonumber 
\end{align}
and
\begin{align}
 P_{0}^{Sv} (b c|y z) & =  \Bigg\{ \frac{\sqrt{2} \mu}{4} \big( P^{0010}_D+P^{0111}_D+P^{1000}_D+P^{1101}_D \big) \nonumber \\
& + \frac{1- \sqrt{2} \mu}{4} \big( P_D^{1000} + P_D^{1100} + P_D^{1101} + P_D^{1001} \big) \Bigg\} \nonumber \\
& := \bordermatrix{
\frac{bc}{yz} & 00 & 01 & 10 & 11 \cr
00 & \frac{1+ \sqrt{2} \mu}{4} & \frac{1- \sqrt{2} \mu}{4} & \frac{1- \sqrt{2} \mu}{4} & \frac{1+ \sqrt{2} \mu}{4} \cr
01 & \frac{1}{4} & \frac{1}{4} & \frac{1}{4} & \frac{1}{4} \cr
10 & \frac{1}{4} & \frac{1}{4} & \frac{1}{4} & \frac{1}{4} \cr
11 & \frac{1- \sqrt{2} \mu}{4} & \frac{1+ \sqrt{2} \mu}{4} & \frac{1+ \sqrt{2} \mu}{4} & \frac{1- \sqrt{2} \mu}{4} }, \nonumber
\end{align}
where each  row and column  corresponds to a fixed  measurement $(yz)$
and a fixed outcome $(bc)$  respectively \footnote{Throughout the paper we will
follow the same convention},

$P_{1} ^{Sv}(b c|y z)  = \begin{pmatrix}
\frac{1- \sqrt{2} \mu}{4} && \frac{1+ \sqrt{2} \mu}{4} && \frac{1+ \sqrt{2} \mu}{4} && \frac{1- \sqrt{2} \mu}{4}\\
\frac{1}{4} && \frac{1}{4} && \frac{1}{4} && \frac{1}{4} \\
\frac{1}{4} && \frac{1}{4} && \frac{1}{4} && \frac{1}{4} \\
\frac{1+ \sqrt{2} \mu}{4} && \frac{1- \sqrt{2} \mu}{4} && \frac{1- \sqrt{2} \mu}{4} && \frac{1+ \sqrt{2} \mu}{4}\\
\end{pmatrix} $, 

$ P_{2} ^{Sv}(b c|y z)  = \begin{pmatrix}
\frac{1}{4} && \frac{1}{4} && \frac{1}{4} && \frac{1}{4} \\
\frac{1+ \sqrt{2} \mu}{4} && \frac{1- \sqrt{2} \mu}{4} && \frac{1- \sqrt{2} \mu}{4} && \frac{1+ \sqrt{2} \mu}{4}\\
\frac{1+ \sqrt{2} \mu}{4} && \frac{1- \sqrt{2} \mu}{4} && \frac{1- \sqrt{2} \mu}{4} && \frac{1+ \sqrt{2} \mu}{4}\\
\frac{1}{4} && \frac{1}{4} && \frac{1}{4} && \frac{1}{4} \\
\end{pmatrix} $, 

$ P_{3}^{Sv} (b c|y z)  = \begin{pmatrix}
\frac{1}{4} && \frac{1}{4} && \frac{1}{4} && \frac{1}{4} \\
\frac{1- \sqrt{2} \mu}{4} && \frac{1+ \sqrt{2} \mu}{4} && \frac{1+ \sqrt{2} \mu}{4} && \frac{1- \sqrt{2} \mu}{4}\\
\frac{1- \sqrt{2} \mu}{4} && \frac{1+ \sqrt{2} \mu}{4} && \frac{1+ \sqrt{2} \mu}{4} && \frac{1- \sqrt{2} \mu}{4}\\
\frac{1}{4} && \frac{1}{4} && \frac{1}{4} && \frac{1}{4} \\
\end{pmatrix} $. \\

Note that, in the range $0<  \mu \leq \frac{1}{\sqrt{2}}$, 
each of the  $P_{\lambda} ^{Sv}(b c|y  z)$ given above belongs to 
BB84 family \cite{GBS16} defined as
\begin{equation}
\label{bb84}
P_{BB84}(bc|yz) = \frac{1 + (-1)^{b \oplus c \oplus y \cdot z} \delta_{y,z} V }{4},
\end{equation}
with $0< V=\sqrt{2}\mu \le 1 $,  upto local reversible operations.
The above family is  local,   as   it   satisfies   the   complete   set   of   Bell-CHSH
(Bell-Clauser-Horne-Shimony-Holt)   inequalities   \cite{CHS+69,
  WW01}.  In Ref. \cite{DBD+17}, it has been demonstrated that the BB84 family cannot be 
reproduced by  shared classical randomness of  dimension $d_\lambda \le 3$.
On the other hand, the BB84 family 
can be reproduced by performing appropriate quantum measurements on $2 \otimes 2$ quantum states \cite{GBS16}.
Hence, each  of the  $P_{\lambda}
^{Sv}(b   c|y  z)$   is  superlocal   in  the   range  $0<   \mu  \leq
\frac{1}{\sqrt{2}}$.

For the fully  local  Svetlichny  family  ($0<  \mu  \leq
\frac{1}{\sqrt{2}}$),  the decomposition (\ref{Svlhv}) defines a classical simulation 
protocol  with different deterministic distributions $P_{\lambda} ^{Sv}(a|x)$ at Alice's side, 
where Alice shares hidden variable $\lambda_{A-BC}$ of dimension $4$ with Bob-Charlie as in Fig. \ref{gsl}. Decomposition (\ref{Svlhv}) represents 2-local form across the bipartite cut $(A|BC)$ of the fully  local  Svetlichny  family  ($0<  \mu  \leq
\frac{1}{\sqrt{2}}$). 

\section{Demonstrating impossibility to reproduce fully local noisy Svetlichny family
in the scenario where Alice preshares the hidden variable $\lambda_{A-BC}$ of dimension $2$
with Bob-Charlie  and uses different deterministic strategies for each $\lambda_{A-BC}$.} \label{SvF2lhv}

Let us try to reproduce the fully local noisy Svetlichny family 
$P^{\mu}_{SvF}$ ($0 < \mu \leq \frac{1}{\sqrt{2}}$) in the scenario as in Fig. \ref{gsl} 
where Alice preshares the hidden variable $\lambda_{A-BC}$ of dimension $2$ with Bob-Charlie and she uses different deterministic strategies for each $\lambda_{A-BC}$. 
In this case, we assume that the fully local noisy Svetlichny family 
$P^{\mu}_{SvF}$ ($0 < \mu \leq \frac{1}{\sqrt{2}}$) can be decomposed in the following way:
\begin{equation}
P^{\mu}_{SvF} = \sum_{\lambda=0}^{1} p_{\lambda} P_{\lambda}^{Sv} (a|x) P_{\lambda} ^{Sv}(b c|y z). \label{2dimSvF}
\end{equation}
Here, $p_0=x_0$, $p_1=x_1$ ($0 <x_0<1$, $0 <x_1<1$, $x_0+x_1 =1$). Since Alice's strategy 
is deterministic one, each of the two probability distributions $P_{0}^{Sv} (a|x)$ and $P_{1}^{Sv} (a|x)$ 
must be equal to any one among $P_D^{00}$, $P_D^{01}$, $P_D^{10}$ and $P_D^{11}$. 
In order to satisfy the marginal probabilities for Alice $P_{SvF}(a|x)$ = $\frac{1}{2}$ $\forall a, x$, 
the only two possible choices of $P_{\lambda} ^{Sv}(a|x)$  are:\\
1) $P_D^{00}$ and $P_D^{01}$ with $x_0=x_1=\frac{1}{2}$\\
2) $P_D^{10}$ and $P_D^{11}$ with $x_0=x_1=\frac{1}{2}$.\\

In case of the first choice, let us assume that 
$P_{0} ^{Sv}(a|x) = P_D^{00}$, $P_{1}^{Sv} (a|x) = P_D^{01}$; $P_{0} ^{Sv}(b c|y z)$ 
and $P_{1} ^{Sv}(b c|y z)$ are given by,

$P_{0} ^{Sv}(b c|y z)  := \begin{pmatrix}
u_{11} && u_{12} && u_{13} && u_{14}\\
u_{21} && u_{22} && u_{23} && u_{24} \\
u_{31} && u_{32} && u_{33} && u_{34}\\
u_{41} && u_{42} && u_{43} && u_{44}\\
\end{pmatrix} $, \\

where $0 \leq u_{ij} \leq 1 \forall i,j$, and $ \sum_{j} u_{ij} =1 \forall i$, 
and

$P_{1} ^{Sv}(b c|y z)  := \begin{pmatrix}
w_{11} && w_{12} && w_{13} && w_{14}\\
w_{21} && w_{22} && w_{23} && w_{24} \\
w_{31} && w_{32} && w_{33} && w_{34}\\
w_{41} && w_{42} && w_{43} && w_{44}\\
\end{pmatrix} $, \\

where $0 \leq w_{ij} \leq 1 \forall i,j$, and $ \sum_{j} w_{ij} =1 \forall i$. \\

Now, with this choice, the box $P^{\mu}_{SvF}$ given by the model   (\ref{2dimSvF})
has

\begin{align}
P^{\mu}_{SvF} & = \bordermatrix{
\frac{abc}{xyz} & 000 & 001 & 010 & 011 & 100 & 101 & 110 & 111 \cr
000 & \frac{u_{11}}{2} &  \frac{u_{12}}{2} &  \frac{u_{13}}{2} &  \frac{u_{14}}{2} &  \frac{w_{11}}{2} &  \frac{w_{12}}{2} &  \frac{w_{13}}{2} &  \frac{w_{14}}{2}\cr
001 & \frac{u_{21}}{2} & \frac{u_{22}}{2} & \frac{u_{23}}{2} & \frac{u_{24}}{2} & \frac{w_{21}}{2} & \frac{w_{22}}{2} & \frac{w_{23}}{2} & \frac{w_{24}}{2} \cr
010 & \frac{u_{31}}{2} & \frac{u_{32}}{2} & \frac{u_{33}}{2} & \frac{u_{34}}{2} & \frac{w_{31}}{2} & \frac{w_{32}}{2} & \frac{w_{33}}{2} & \frac{w_{34}}{2} \cr
011 & \frac{u_{41}}{2} & \frac{u_{42}}{2} & \frac{u_{43}}{2} & \frac{u_{44}}{2} & \frac{w_{41}}{2} & \frac{w_{42}}{2} & \frac{w_{43}}{2} & \frac{w_{44}}{2} \cr 
100 & \frac{u_{11}}{2} &  \frac{u_{12}}{2} &  \frac{u_{13}}{2} &  \frac{u_{14}}{2} &  \frac{w_{11}}{2} &  \frac{w_{12}}{2} &  \frac{w_{13}}{2} &  \frac{w_{14}}{2}\cr
101 & \frac{u_{21}}{2} & \frac{u_{22}}{2} & \frac{u_{23}}{2} & \frac{u_{24}}{2} & \frac{w_{21}}{2} & \frac{w_{22}}{2} & \frac{w_{23}}{2} & \frac{w_{24}}{2} \cr
110 & \frac{u_{31}}{2} & \frac{u_{32}}{2} & \frac{u_{33}}{2} & \frac{u_{34}}{2} & \frac{w_{31}}{2} & \frac{w_{32}}{2} & \frac{w_{33}}{2} & \frac{w_{34}}{2} \cr
111 & \frac{u_{41}}{2} & \frac{u_{42}}{2} & \frac{u_{43}}{2} & \frac{u_{44}}{2} & \frac{w_{41}}{2} & \frac{w_{42}}{2} & \frac{w_{43}}{2} & \frac{w_{44}}{2} }, 
\label{b2}
\end{align}
where each row and column corresponds to a fixed measurement $(xyz)$ and a fixed outcome $(abc)$ respectively. \\

From Eq. (\ref{b2}), it can be seen that
\be
P^{\mu}_{SvF} (a b c|001) = P^{\mu}_{SvF} (a b c|101), \nonumber
\ee
which is not true for the fully local Svetlichny family as given in Eq. (\ref{SvF})
with $0 < \mu \leq \frac{1}{\sqrt{2}}$. Because, 
in case of fully local $P^{\mu}_{SvF}$ ($0 < \mu \leq \frac{1}{\sqrt{2}}$) given in Eq. (\ref{SvF}), 
\begin{equation}
P^{\mu}_{SvF} (a b c|001)  = \frac{2+(-1)^{a\oplus b \oplus c }\sqrt{2}\mu}{16} \nonumber
\end{equation}
and 
\begin{equation}
P^{\mu}_{SvF} (a b c|101) = \frac{2+(-1)^{a\oplus b \oplus c \oplus 1}\sqrt{2}\mu}{16}. \nonumber
\end{equation}

Again, from Eq. (\ref{b2}), 
it can be seen that 
\be
P^{\mu}_{SvF} (a b c|010) = P^{\mu}_{SvF} (a b c|110), \nonumber 
\ee
which is not true for the fully local Svetlichny family as given in Eq. (\ref{SvF})
with $0 < \mu \leq \frac{1}{\sqrt{2}}$. Because, in case of fully local $P^{\mu}_{SvF}$ ($0 < \mu \leq \frac{1}{\sqrt{2}}$)
given in Eq. (\ref{SvF}), 
\begin{equation}
P^{\mu}_{SvF} (a b c|010)  = \frac{2+(-1)^{a\oplus b \oplus c }\sqrt{2}\mu}{16} \nonumber
\end{equation}
and
\begin{equation}
P^{\mu}_{SvF} (a b c|110) = \frac{2+(-1)^{a\oplus b \oplus c \oplus 1}\sqrt{2}\mu}{16}. \nonumber
\end{equation}

Hence, in this case,
though the marginal probabilities for Alice $P_{SvF}(a|x)$ are satisfied, 
all the tripartite joint probability distributions $P_{SvF}^{\mu}$ are not satisfied simultaneously.\\

Similarly, in case of the first choice, if we assume that 
$P_{0} ^{Sv}(a|x) = P_D^{01}$, $P_{1}^{Sv} (a|x) = P_D^{00}$, 
then the marginal probabilities for Alice $P_{SvF}(a|x)$ are satisfied, 
but all the tripartite joint probability distributions $P_{SvF}^{\mu}$ are not satisfied simultaneously.\\

Now, in case of the second choice, let us assume that
$P_{0} ^{Sv}(a|x) = P_D^{10}$, $P_{1}^{Sv} (a|x) = P_D^{11}$; 
$P_{0} ^{Sv}(b c|y z)$ and $P_{1} ^{Sv}(b c|y z)$ are given by,\\

$P_{0} ^{Sv}(b c|y z)  = \begin{pmatrix}
u^{'}_{11} && u^{'}_{12} && u^{'}_{13} && u^{'}_{14}\\
u^{'}_{21} && u^{'}_{22} && u^{'}_{23} && u^{'}_{24} \\
u^{'}_{31} && u^{'}_{32} && u^{'}_{33} && u^{'}_{34}\\
u^{'}_{41} && u^{'}_{42} && u^{'}_{43} && u^{'}_{44}\\
\end{pmatrix} $, \\

where $0 \leq u^{'}_{ij} \leq 1 \forall i,j$, and $ \sum_{j} u^{'}_{ij} =1 \forall i$, 
and\\

$P_{1} ^{Sv}(b c|y z)  = \begin{pmatrix}
w^{'}_{11} && w^{'}_{12} && w^{'}_{13} && w^{'}_{14}\\
w^{'}_{21} && w^{'}_{22} && w^{'}_{23} && w^{'}_{24} \\
w^{'}_{31} && w^{'}_{32} && w^{'}_{33} && w^{'}_{34}\\
w^{'}_{41} && w^{'}_{42} && w^{'}_{43} && w^{'}_{44}\\
\end{pmatrix} $, \\

where $0 \leq w^{'}_{ij} \leq 1 \forall i,j$, and $ \sum_{j} w^{'}_{ij} =1 \forall i$. 

Now, with this choice, the box $P^{\mu}_{SvF}$ ($0 < \mu \leq \frac{1}{\sqrt{2}}$) given by the model (\ref{2dimSvF}) 
has,

\begin{align}
P^{\mu}_{SvF}  = \bordermatrix{
\frac{abc}{xyz} & 000 & 001 & 010 & 011 & 100 & 101 & 110 & 111 \cr
000 & \frac{u^{'}_{11}}{2} &  \frac{u^{'}_{12}}{2} &  \frac{u^{'}_{13}}{2} &  \frac{u^{'}_{14}}{2} &  \frac{w^{'}_{11}}{2} &  \frac{w^{'}_{12}}{2} &  \frac{w^{'}_{13}}{2} &  \frac{w^{'}_{14}}{2}\cr
001 & \frac{u^{'}_{21}}{2} & \frac{u^{'}_{22}}{2} & \frac{u^{'}_{23}}{2} & \frac{u^{'}_{24}}{2} & \frac{w^{'}_{21}}{2} & \frac{w^{'}_{22}}{2} & \frac{w^{'}_{23}}{2} & \frac{w^{'}_{24}}{2} \cr
010 & \frac{u^{'}_{31}}{2} & \frac{u^{'}_{32}}{2} & \frac{u^{'}_{33}}{2} & \frac{u^{'}_{34}}{2} & \frac{w^{'}_{31}}{2} & \frac{w^{'}_{32}}{2} & \frac{w^{'}_{33}}{2} & \frac{w^{'}_{34}}{2} \cr
011 & \frac{u^{'}_{41}}{2} & \frac{u^{'}_{42}}{2} & \frac{u^{'}_{43}}{2} & \frac{u^{'}_{44}}{2} & \frac{w^{'}_{41}}{2} & \frac{w^{'}_{42}}{2} & \frac{w^{'}_{43}}{2} & \frac{w^{'}_{44}}{2} \cr 
100 & \frac{w^{'}_{11}}{2} &  \frac{w^{'}_{12}}{2} &  \frac{w^{'}_{13}}{2} &  \frac{w^{'}_{14}}{2} & \frac{u^{'}_{11}}{2} &  \frac{u^{'}_{12}}{2} &  \frac{u^{'}_{13}}{2} &  \frac{u^{'}_{14}}{2} \cr
101 & \frac{w^{'}_{21}}{2} & \frac{w^{'}_{22}}{2} & \frac{w^{'}_{23}}{2} & \frac{w^{'}_{24}}{2} & \frac{u^{'}_{21}}{2} & \frac{u^{'}_{22}}{2} & \frac{u^{'}_{23}}{2} & \frac{u^{'}_{24}}{2} \cr
110 & \frac{w^{'}_{31}}{2} & \frac{w^{'}_{32}}{2} & \frac{w^{'}_{33}}{2} & \frac{w^{'}_{34}}{2} & \frac{u^{'}_{31}}{2} & \frac{u^{'}_{32}}{2} & \frac{u^{'}_{33}}{2} & \frac{u^{'}_{34}}{2}  \cr
111 &  \frac{w^{'}_{41}}{2} & \frac{w^{'}_{42}}{2} & \frac{w^{'}_{43}}{2} & \frac{w^{'}_{44}}{2} & \frac{u^{'}_{41}}{2} & \frac{u^{'}_{42}}{2} & \frac{u^{'}_{43}}{2} & \frac{u^{'}_{44}}{2} }.
\label{b22}
\end{align}

From Eq. (\ref{b22}), it can be seen that 
\be
P^{\mu}_{SvF} (0 b c|000) = P^{\mu}_{SvF} (1 b c|100), \nonumber
\ee
which is not true for the fully local Svetlichny family as given in Eq. (\ref{SvF})
with $0 < \mu \leq \frac{1}{\sqrt{2}}$.
Because, in case of fully local $P^{\mu}_{SvF}$ ($0 < \mu \leq \frac{1}{\sqrt{2}}$) given in Eq. (\ref{SvF}), 
\begin{equation}
P^{\mu}_{SvF} (0 b c|000)  = \frac{2+(-1)^{b \oplus c }\sqrt{2}\mu}{16} \nonumber
\end{equation}
and
\begin{equation}
P^{\mu}_{SvF} (1 b c|100) = \frac{2+(-1)^{b \oplus c  \oplus 1}\sqrt{2}\mu}{16}. \nonumber
\end{equation}

From Eq. (\ref{b22}), it can be seen that 
\be
P^{\mu}_{SvF} (1 b c|000) = P^{\mu}_{SvF} (0 b c|100), \nonumber
\ee
which is not true for the fully local Svetlichny family as given in Eq. (\ref{SvF}) 
for $0 < \mu \leq \frac{1}{\sqrt{2}}$. Because, in case of fully local $P^{\mu}_{SvF}$ 
($0 < \mu \leq \frac{1}{\sqrt{2}}$) as given in Eq. (\ref{SvF}), 
\begin{equation}
P^{\mu}_{SvF} (1 b c|000)  = \frac{2+(-1)^{b \oplus c   \oplus 1}\sqrt{2}\mu}{16} \nonumber
\end{equation}
and
\begin{equation}
P^{\mu}_{SvF} (0 b c|100) = \frac{2+(-1)^{b \oplus c}\sqrt{2}\mu}{16}. \nonumber
\end{equation}

Again, from Eq. (\ref{b22}), it can be seen that 
\be
P^{\mu}_{SvF} (0 b c|011) = P^{\mu}_{SvF} (1 b c|111), \nonumber
\ee
which is not true for fully local $P^{\mu}_{SvF}$ as given in 
Eq. (\ref{SvF}) for $0 < \mu \leq \frac{1}{\sqrt{2}}$. Because, in case of fully local 
$P^{\mu}_{SvF}$ ($0 < \mu \leq \frac{1}{\sqrt{2}}$) as given in Eq. (\ref{SvF}), 
\begin{equation}
P^{\mu}_{SvF} (0 b c|011)  = \frac{2+(-1)^{b \oplus c \oplus 1}\sqrt{2}\mu}{16} \nonumber
\end{equation}
and, 
\begin{equation}
P^{\mu}_{SvF} (1 b c|111) = \frac{2+(-1)^{b \oplus c}\sqrt{2}\mu}{16} \nonumber
\end{equation}

From Eq. (\ref{b22}), it can be seen that 
\be
P^{\mu}_{SvF} (1 b c|011) = P^{\mu}_{SvF} (0 b c|111),
\ee
which is not true for fully local $P^{\mu}_{SvF}$ as given in Eq. (\ref{SvF})
for $0 < \mu \leq \frac{1}{\sqrt{2}}$. Because, in case of fully local 
$P^{\mu}_{SvF}$ ($0 < \mu \leq \frac{1}{\sqrt{2}}$) as given in Eq. (\ref{SvF}), 
\begin{equation}
P^{\mu}_{SvF} (1 b c|011)  = \frac{2+(-1)^{b \oplus c}\sqrt{2}\mu}{16} \nonumber
\end{equation}
and
\begin{equation}
P^{\mu}_{SvF} (0 b c|111) = \frac{2+(-1)^{b \oplus c   \oplus 1}\sqrt{2}\mu}{16} \nonumber.
\end{equation}

Hence, in this case, though the marginal probabilities for Alice $P_{SvF}(a|x)$ are satisfied, 
all the tripartite joint probability distributions $P_{SvF}^{\mu}$ are not satisfied simultaneously.\\

Similarly, in case of the second choice, if we assume that $P_{0} ^{Sv}(a|x) = P_D^{11}$, 
$P_{1}^{Sv} (a|x) = P_D^{10}$, then the marginal probabilities for Alice $P_{SvF}(a|x)$ are satisfied, 
but all the tripartite joint probability distributions $P_{SvF}^{\mu}$ are not satisfied simultaneously.\\

Note that this proof is valid for any two bipartite correlations 
$P_{\lambda} ^{Sv}(b c|y z)$ ($\lambda = 0,1$) shared between Bob and Charlie at each $\lambda_{A-BC}$  
without any constraint on the correlations. Hence, it is obvious that this proof will also 
be valid when these two bipartite correlations $P_{\lambda} ^{Sv}(b c|y z)$ ($\lambda = 0,1$) 
at Bob-Charlie's side are NS boxes and satisfy the marginal probabilities 
$P_{SvF}(b|y)$, $P_{SvF}(c|z)$ at Bob and Charlie's side, respectively.\\

It is, therefore, impossible to reproduce the fully local box 
$P^{\mu}_{SvF}$ ($0 < \mu \leq \frac{1}{\sqrt{2}}$) in the scenario as in Fig. \ref{gsl} 
where Alice preshares the hidden variable $\lambda_{A-BC}$ of dimension $2$ with Bob-Charlie.

\section{$2$-local form across the bipartition $(A|BC)$ for the fully  local  Mermin family  $P^{\nu}_{MF}$  in the range $0 < \nu\le\frac{1}{2}$} \label{2lfm}

For $0 < \nu\le\frac{1}{2}$ the fully local noisy Mermin box can be
decomposed in a convex mixture of the 3-local deterministic boxes.
In this range, we consider the following decomposition for the 
Mermin family
in terms of the 3-local
deterministic boxes: 

\begin{align} 
P^{\nu}_{MF}&=\frac{1}{4} P_D^{00} \Bigg\{ \frac{2 \nu}{8} \Big( 
P^{0000}_D+P^{0010}_D+P^{0101}_D+P^{0111}_D \nonumber \\
& + P^{1000}_D+P^{1011}_D+P^{1101}_D+P^{1110}_D \Big) \nonumber \\
&+ \frac{1- 2 \nu}{4} \Big( P_D^{1000} + P_D^{1100} + P_D^{1101} + P_D^{1001} \Big) \Bigg\} \nonumber \\
& + \frac{1}{4} P_D^{01} \Bigg\{ \frac{2 \nu}{8} \Big( 
P^{0001}_D+P^{0011}_D+P^{0100}_D+P^{0110}_D \nonumber \\
& + P^{1001}_D+P^{1010}_D+P^{1100}_D+P^{1111}_D \Big) \nonumber \\
&+ \frac{1- 2 \nu}{4} \Big( P_D^{1000} + P_D^{1100} + P_D^{1101} + P_D^{1001} \Big) \Bigg\} \nonumber \\
& + \frac{1}{4} P_D^{01} \Bigg\{ \frac{2 \nu}{8} \Big( 
P^{0000}_D+P^{0011}_D+P^{0101}_D+P^{0110}_D \nonumber \\ 
& + P^{1100}_D+P^{1001}_D+P^{1110}_D+P^{1011}_D \Big) \nonumber \\
&+ \frac{1- 2 \nu}{4} \Big( P_D^{1000} + P_D^{1100} + P_D^{1101} + P_D^{1001} \Big) \Bigg\} \nonumber \\
& + \frac{1}{4} P_D^{01} \Bigg\{ \frac{2 \nu}{8} \Big( 
P^{0001}_D+P^{0010}_D+P^{0100}_D+P^{0111}_D \nonumber \\ 
& + P^{1000}_D+P^{1010}_D+P^{1101}_D+P^{1111}_D \Big) \nonumber \\
&+ \frac{1- 2 \nu}{4} \Big( P_D^{1000} + P_D^{1100} + P_D^{1101} + P_D^{1001} \Big) \Bigg\} \nonumber \\
& := \sum_{\lambda=0}^{3} r_{\lambda} P_{\lambda}^{M} (a|x) P_{\lambda} ^{M}(b c|y z), \label{mlhv}
\end{align}

where $r_0$ = $r_1$ = $r_2$ = $r_3$ = $\frac{1}{4}$; 
\begin{align}
&P_{0} ^{M}(a|x) = P_D^{00}, P_{1}^{M} (a|x) = P_D^{01}, \nonumber \\
&P_{2}^{M}(a|x) = P_D^{10}, P_{3}^{M}(a|x) = P_D^{11}; \nonumber
\end{align}
and\\

\begin{align}
 P_{0}^{M} (b c|y z) & =  \Bigg\{ \frac{2 \nu}{8} \Big( 
P^{0000}_D+P^{0010}_D+P^{0101}_D+P^{0111}_D \nonumber \\
& + P^{1000}_D+P^{1011}_D+P^{1101}_D+P^{1110}_D \Big) \nonumber \\
&+ \frac{1- 2 \nu}{4} \Big( P_D^{1000} + P_D^{1100} + P_D^{1101} + P_D^{1001} \Big) \Bigg\} \nonumber \\
& = \begin{pmatrix}
 \frac{1+\nu}{4} && \frac{1-\nu}{4} && \frac{1-\nu}{4} && \frac{1+\nu}{4} \\
 \frac{1+\nu}{4} && \frac{1-\nu}{4} && \frac{1-\nu}{4} && \frac{1+\nu}{4} \\
 \frac{1+\nu}{4} && \frac{1-\nu}{4} && \frac{1-\nu}{4} && \frac{1+\nu}{4} \\
 \frac{1-\nu}{4} && \frac{1+\nu}{4} && \frac{1+\nu}{4} && \frac{1-\nu}{4}  \nonumber\\
\end{pmatrix}, \\
P_{1} ^{M}(b c|y z)  &= \begin{pmatrix}
\frac{1-\nu}{4} && \frac{1+\nu}{4} && \frac{1+\nu}{4} && \frac{1-\nu}{4}\\
\frac{1-\nu}{4} && \frac{1+\nu}{4} && \frac{1+\nu}{4} && \frac{1-\nu}{4}\\
\frac{1-\nu}{4} && \frac{1+\nu}{4} && \frac{1+\nu}{4} && \frac{1-\nu}{4}\\
\frac{1+\nu}{4} && \frac{1-\nu}{4} && \frac{1-\nu}{4} && \frac{1+\nu}{4} \nonumber \\
\end{pmatrix}, \\
 P_{2} ^{M}(b c|y z)  &= \begin{pmatrix}
\frac{1-\nu}{4} && \frac{1+\nu}{4} && \frac{1+\nu}{4} && \frac{1-\nu}{4}\\
\frac{1+\nu}{4} && \frac{1-\nu}{4} && \frac{1-\nu}{4} && \frac{1+\nu}{4}\\
\frac{1+\nu}{4} && \frac{1-\nu}{4} && \frac{1-\nu}{4} && \frac{1+\nu}{4}\\
\frac{1+\nu}{4} && \frac{1-\nu}{4} && \frac{1-\nu}{4} && \frac{1+\nu}{4} \nonumber \\
\end{pmatrix}, \\
P_{3}^{M} (b c|y z) & = \begin{pmatrix}
\frac{1+\nu}{4} && \frac{1-\nu}{4} && \frac{1-\nu}{4} && \frac{1+\nu}{4}\\
\frac{1-\nu}{4} && \frac{1+\nu}{4} && \frac{1+\nu}{4} && \frac{1-\nu}{4}\\
\frac{1-\nu}{4} && \frac{1+\nu}{4} && \frac{1+\nu}{4} && \frac{1-\nu}{4}\\
\frac{1-\nu}{4} && \frac{1+\nu}{4} && \frac{1+\nu}{4} && \frac{1-\nu}{4} \nonumber \\
\end{pmatrix}.  
\end{align}

Note that, each of the probability distributions $P_{\lambda}^{M} (b c|y z)$ is local 
for  $0< \nu \leq \frac{1}{2}$, as they satisfy the complete set of 
Bell-CHSH inequalities \cite{CHS+69, WW01} in this range. In fact, each of 
the $P_{\lambda}^{M} (b c|y z)$ is superlocal in the range $0< \nu \leq \frac{1}{2}$. 
Because each  of the  $P_{\lambda}
^{M}(b   c|y  z)$ belongs to CHSH family \cite{GBS16},
\begin{equation}
P_{CHSH}(bc|yz)       =       \frac{2+(-1)^{b\oplus       c\oplus
    yz}\sqrt{2}V}{8}, \label{chshfam}
\end{equation}
with $0 < V=\sqrt{2}\nu \leq 1$,  upto local reversible operations.
In Ref. \cite{JAS16}, 
it has been shown that the above local box is superlocal for any $V>0$.

For the fully  local  noisy   Mermin-box  ($0< \nu \leq \frac{1}{2}$),  the decomposition (\ref{mlhv}) defines a classical simulation 
protocol with different deterministic distributions $P_{\lambda} ^{M}(b c|y z)$ at Alice's side,
where Alice shares hidden variable $\lambda_{A-BC}$ of dimension $4$ as in Fig. \ref{gsl}.  Decomposition (\ref{mlhv}) represents 2-local form across the bipartite cut $(A|BC)$ of the fully  local  noisy   Mermin-box  ($0< \nu \leq \frac{1}{2}$). 

\section{$2$-local form across the bipartition $(A|BC)$ for the Mermin family  $P^{\nu}_{MF}$  in the range $0 < \nu\le1$} \label{2lfm2}

For  $0 <  \nu\le 1$  the 2-local 
Mermin  family can  be decomposed  in a  convex mixture  of the  local
vertices and 2-local vertices as follows:
\begin{align}
P^{0<\nu<1}_{MF}&=\frac{1}{4}P^{00}_D(a|x) \left(\nu P_{PR}^{000}(bc|yz)+ \frac{1-\nu}{16} \sum_{\gamma, \epsilon, \zeta, \eta}  P^{\gamma\epsilon\zeta\eta}_D(bc|yz) \right) \nonumber \\
&+\frac{1}{4}P^{01}_D(a|x) \left(\nu P_{PR}^{001}(bc|yz)+\frac{1-\nu}{16} \sum_{\gamma, \epsilon, \zeta, \eta}  P^{\gamma\epsilon\zeta\eta}_D(bc|yz) \right) \nonumber \\
&+\frac{1}{4}P^{10}_D(a|x) \left(\nu P_{PR}^{111}(bc|yz)+\frac{1-\nu}{16} \sum_{\gamma, \epsilon, \zeta, \eta}  P^{\gamma\epsilon\zeta\eta}_D(bc|yz) \right)\nonumber \\
&+\frac{1}{4}P^{11}_D(a|x) \left(\nu P_{PR}^{110}(bc|yz)+\frac{1-\nu}{16} \sum_{\gamma, \epsilon, \zeta, \eta}  P^{\gamma\epsilon\zeta\eta}_D(bc|yz) \right) \nonumber \\
&: = \sum_{\lambda=0}^{3} r_{\lambda} P_{\lambda}^{M} (a|x) P_{\lambda} ^{M}(b c|y z), \label{mlhv3}
\end{align}
where $r_0$ = $r_1$ = $r_2$ = $r_3$ = $\frac{1}{4}$; 
\begin{align}
&P_{0} ^{M}(a|x) = P_D^{00}, P_{1}^{M} (a|x) = P_D^{01}, \nonumber \\
&P_{2}^{M}(a|x) = P_D^{10}, P_{3}^{M}(a|x) = P_D^{11}; \nonumber
\end{align}
and\\

\begin{align}
 P_{0}^{M} (b c|y z) & =  \nu P_{PR}^{000}(bc|yz)+\frac{1-\nu}{16} \sum_{\gamma, \epsilon, \zeta, \eta}  P^{\gamma\epsilon\zeta\eta}_D(bc|yz),  \nonumber \\
 P_{1}^{M} (b c|y z) & =  \nu P_{PR}^{001}(bc|yz)+ \frac{1-\nu}{16} \sum_{\gamma, \epsilon, \zeta, \eta}  P^{\gamma\epsilon\zeta\eta}_D(bc|yz),  \nonumber \\
 P_{2}^{M} (b c|y z) & = \nu P_{PR}^{111}(bc|yz)+ \frac{1-\nu}{16} \sum_{\gamma, \epsilon, \zeta, \eta}  P^{\gamma\epsilon\zeta\eta}_D(bc|yz),  \nonumber \\
 P_{3}^{M} (b c|y z) & =  \nu P_{PR}^{110}(bc|yz)+ \frac{1-\nu}{16} \sum_{\gamma, \epsilon, \zeta, \eta}  P^{\gamma\epsilon\zeta\eta}_D(bc|yz).  \nonumber
\end{align}

Note that, each of the probability distributions $P_{\lambda}^{M} (b c|y z)$ satisfy NS principle 
for  $0< \nu \leq 1$,  Hence, the noisy Mermin box has the 2-local form across the bipartite cut $(A|BC)$ in the range $0< \nu \leq 1$. For the   $2$-local  noisy   Mermin-box  ($0<  \nu  \leq
1$),  the decomposition (\ref{mlhv3}) defines a classical simulation 
protocol as in Fig. \ref{sl} where Alice shares hidden variable
$\lambda_{A-BC}$ of dimension $4$.

\end{document}